\documentclass[11pt,a4paper]{article}
\usepackage{a4,amsfonts,amsmath,amssymb,enumerate,graphicx,harvard,longtable,lscape,subfigure,theorem}
\usepackage[english]{babel}

\newtheorem{theorem}{Theorem}[section]
\newtheorem{definition}[theorem]{Definition}

\newtheorem{lemma}
[theorem]{Lemma}{\theorembodyfont{\rmfamily}}{\theorembodyfont{\rmfamily}}

\newenvironment{proof}
{\begin{trivlist}\item[]{{\sc Proof.}}}{\hfill{$\square$}\noindent\end{trivlist}}

\newcommand{\od}{\operatorname{d}}
\newcommand{\R}{\mathsf{R}}
\newcommand{\V}{\mathsf{V}}

\begin{document}
\title{\textbf{Representation-Compatible Power Indices}}

\author{
Serguei Kaniovski\footnote{
Austrian Institute of Economic Research (WIFO), Austria. E-mail: {serguei.kaniovski@wifo.ac.at}.}
\and
Sascha Kurz\footnote{
Department of Mathematics, University of Bayreuth, Germany. E-mail: {sascha.kurz@uni-bayreuth.de}.}
}
\date{\today}
\maketitle

\begin{abstract}
This paper studies power indices based on average representations of a weighted game. If restricted to 
account for the lack of power of dummy voters, average representations become coherent measures of voting 
power, with power distributions being proportional to the distribution of weights in the average 
representation. This makes these indices representation-compatible, a property not fulfilled by 
classical power indices. Average representations can be tailored to reveal the equivalence classes of
voters defined by the Isbell desirability relation, which leads to a pair of new power indices that 
ascribes equal power to all members of an equivalence class.



\end{abstract}

\bigskip
\bigskip

\section{Introduction}\label{sec:intro}
We commonly represent a weighted voting game using an integer-valued vector of voting weights 
$(w_1,w_2,\dots,w_n)$ and an integer quota $q$. The vector of weights conveys the number of 
votes each of the $n$ voters commands. The game involves each voter casting all her votes as 
a bloc either for or against a motion. The motion is passed if the total number of votes cast 
by the voters in favor of the motion is greater than or equal to $q$; otherwise, the motion 
is rejected. In this paper, we introduce two new measures of power for weighted games, and 
study the properties of a family of representation-compatible power indices that now includes 
four measures.

Written as $[q;w_1,w_2,\dots,w_n]$, a representation conveys the set of winning coalitions of 
the weighted game. For example, winning coalitions for a game represented by $[51;47,46,5,2]$ are
$$
\{\{1, 2\}, \{1, 3\}, \{2, 3\}, \{1, 2, 3\}, \{1, 2, 4\}, \{1, 3, 4\}, \{2, 3, 4\}, \{1, 2, 3, 4\}\}.
$$
This set allows us to compute the voting power of the voters. A voter is critical to a winning 
coalition if the coalition becomes a losing one should she withdraw her support. No voter is critical 
in a coalition of all voters $\{1, 2, 3, 4\}$. The first voter is critical in 
$\{1, 2\}$, $\{1, 2, 4\}$, $\{1, 3\}$ and $\{1, 3, 4\}$. The largest three voters acting 
together can pass a motion, but none of them would be critical to the success of the coalition 
$\{1, 2, 3\}$. Since each of the three large voters are critical in exactly four winning coalitions, 
they should be equally powerful. The smallest voter is a dummy (Definition~\ref{def_dummy}), because 
she is irrelevant to the success of any coalition. The above considerations suggest $(1/3,1/3,1/3,0)$ 
as a plausible distribution of voting power. Note that the three voters are equally powerful despite 
them unequal weights. The distribution of weights $(0.47,0.46,0.05,0.02)$ is markedly different from 
the distribution of power. This difference would decrease if we chose to represent the above game 
using $[51;34,33,33,0]$. It would completely disappear if we chose the rational-valued representation 
$[2/3;1/3,1/3,1/3,0]$.

Although the set of winning coalitions uniquely defines a weighted game, there are infinitely 
many representations consistent with any given set of winning coalitions. If we adhere to a total 
of 100 votes among three voters and a quota of 51, then there will be 8924 rational-valued weight 
distributions consistent with the power vector $(1/3,1/3,1/3,0)$. If the quota itself is considered 
to be part of the specification, then there will be 79800 possibilities to represent the game. Scaling 
the quota and weights in any of these representations by the same factor would preserve the set of 
winning coalitions and consequently define the same weighted game.

Since any admissible representation defines the game, the multiplicity of representations has no 
bearing on the definition of a weighted game. The multiplicity poses a problem if we want to design 
a weighted game (a voting institution such as a parliament) with a given distribution of power or to 
compare the distribution of power to the distribution of weights. Whereas imposing additional criteria 
can reduce the number of feasible representations and even identify a unique appropriate representation 
in sufficiently small, weighted games\footnote{We could minimize the total sum of integral voting weights 
to obtain a unique representation for $n\le7$, as in \citeasnoun{freixas_molin1}, or \citeasnoun{freixas_kurz1}. 
For other options, see Application~9.9 in \citeasnoun{crama_hammer1}.}, the multiplicity makes unclear which 
representation the power distribution should be compared to. The average representations can reasonably be 
compared to power distributions of various power indices, as they uniquely summarize the set of admissible 
representations.

The proportionality of power and weight has received significant attention in the literature. It has been 
noted that none of the classical power indices yield power distributions that coincide with weight 
distributions for all weighted games. The observation that the distribution of voting power is different 
from the distribution of voting weight has been one of the motivating factors behind the development of 
the theory of power measurement in weighted games \cite{felsen_mach2}. Two recent theoretical studies 
provide conditions for which the weight distribution and the power distribution coincide. These results 
are available for the Banzhaf index by \citeasnoun{houy_zwicker1} and the nucleolus 
by \citeasnoun{kurz_napel_nohn1}. An exception is the recently introduced Minimum Sum Representation 
Index (MSRI) index by \citeasnoun{freixas_kan1}, which is specifically designed to fulfill proportionality.

The average representations come very close to being valid measures of power for weighted games. If 
restricted to account for the lack of power of dummy voters, average representations become coherent 
measures of voting power, with power distributions being proportional to the distribution of weights 
in the average representation. Restricting the polytope implied by the set of minimal winning and 
maximal losing coalitions yields average representations that are dummy-revealing. These restricted 
average representations satisfy \possessivecite{freixas_gamba1} coherency criteria for power indices, 
which are essentially equivalent to the widely accepted `minimal adequacy postulate' 
by \citeasnoun{felsen_mach1} (p.~222). These indices of voting power, called AWI and ARI, are based 
on restricted average representations that respect proportionality between power and weight \cite{kan_kurz1}.

The above modification of an average representation suggests that we can endow the indices with 
further qualities by tailoring the polytope. Restrictions based on the equivalence classes of voters 
defined by the Isbell desirability relation lead to another pair of power indices that ascribes equal 
power to all members of an equivalence class. The two new measure of voting power introduced in this 
paper are specifically designed to recognize the equivalence of voters. We thus propose two new power 
indices that respect proportionality between power and weight, are dummy-revealing as well as 
type-revealing. Together with AWI and ARI proposed in \citeasnoun{kan_kurz1}, the new type-revealing 
indices, called AWTI and ARTI, complete the family of representation-compatible power measures studied 
in this paper.

In the next section, we recall the preliminaries required to define representation-compatible power 
indices. Section~\ref{sec:definition} defines the indices, verifies their coherency as measures of 
power and discusses computational issues. The computation of representation-compatible indices involves 
the integration of monomials on highly-dimensional polytopes with rational vertices. The main drawback of 
these indices is the computational burden of numerical integration.\footnote{Programs for computing the 
indices are available from the authors upon request.} Section~\ref{sec:properties} first compares the 
power distributions generated by representation-compatible indices to power distributions according to 
the \citeasnoun{banzh1} and \citeasnoun{shapl_shub1} indices in small weighted games, and then discusses 
their vulnerabilities to certain anomalies, commonly referred to as voting paradoxes. 
Section~\ref{sec:integral_weights} discusses some aspects of the integer-valued representations that 
have been recently used by \citeasnoun{freixas_kan1} to construct a power index. It turns out that average 
representations and the MSRI are related. The final section offers concluding remarks and ideas for future 
research.

\section{Notation and preliminaries}\label{sec:notation}
\subsection{Simple games and weighted games}\label{subsec:simplegames}
A (monotonic) simple game is the most general type of binary voting game.
\begin{definition}\label{def_simple_game}
A simple game $v$ is a mapping $v:2^n\rightarrow \{0,1\}$, where $N=\{1,\dots,n\}$ is the set of voters, 
such that $v(\emptyset)=0$, $v(N)=1$, and $v(S)\le v(T)$ for all $S\subseteq T\subseteq N$ (monotonicity). 
\end{definition}
A subset $S\subseteq N$ is called a coalition of $v$. There are $2^n$ such coalitions in a simple game 
with $n$ voters. A coalition $S$ is winning if $v(S)=1$, and losing if $v(S)=0$. The monotonicity ensures 
that enlarging a winning coalition cannot make it a losing one, which is a sensible assumption.

A winning coalition $S$ is called a \emph{minimal winning coalition} if none of its proper subsets are 
winning. Similarly, a losing coalition $T$ is called a \emph{maximal losing coalition} if none of its 
proper supersets are losing. The set of minimal winning coalitions $\mathcal{W}^m$, or the set of maximal 
losing coalitions $\mathcal{L}^m$, uniquely defines a simple game. For the game represented by 
$[51;47,46,5,2]$, the set of minimal winning coalitions is given by $\{\{1,2\},\{1,3\},\{2,3\}\}$. 
We define a simple game using the set of minimal winning coalitions as opposed to the set of winning 
coalitions, as the former definition is more compact.

A weighted game is a simple game that admits a representation $[q;w_1,w_2,\dots,w_n]$.
\begin{definition}\label{def_weighted_majority_game}
A simple game $v$ is weighted, if there exist real numbers $w_1,\dots,w_n\ge 0$ and $q>0$, such that
$$
\sum_{s\in S} w_s\ge q\quad\Longleftrightarrow\quad v(S)=1,
$$
for all $S\subseteq N$. We write: $(N,v)=[q;w_1,\dots,w_n]$.
\end{definition}
In this paper, we consider weighted games, as this type of binary voting games is most relevant to the 
applied power measurement and institutional design. A common institution that uses weighted voting for 
decision making is the shareholder assembly in a corporation. The voting weight of a shareholder equals 
the number of ordinary shares she holds. This example also includes voting by the member states in 
multilateral institutions such as the World Bank and the IMF. In the political arena, voting in 
parliaments can be viewed as a weighted game, provided party discipline is absolute. The frequently 
studied voting in the Council of Ministers of the EU can be viewed, with some simplification of the 
double-majority voting rule stipulated by the Lisbon Treaty, as a weighted game. In the examples above, 
the voting weights are non-negative integers. The conditions required for a simple game to be a weighted 
game have been studied extensively in the literature.\footnote{See, \citeasnoun{taylor_zwick1}. For a 
survey, see Chapter 9.8 in \citeasnoun{crama_hammer1}.}

\subsection{Equivalence classes of voters}\label{subsec:equiv}
The equivalence classes serve two purposes. They partition the set of voters according to their effect 
on the decisiveness of coalitions. Any reasonable measure of voting power should, therefore, recognize 
the equivalence classes. Second, while each majority game has an infinite number of representations, 
the number of possible partitions of all games with a given number of voters is finite. The 
qualifier `for all games' then stands for `all feasible partitions of players in classes'. Our 
comparisons between power indices presented in Section~\ref{sec:properties} were obtained using this 
set of games, where each game is defined by its minimum sum representation.\footnote{See, \citeasnoun{FrPo10TD}.}
\begin{definition}\label{def_equivalence_classes}
Given a simple game $v$, we say that two voters $i,j\in N$ are equivalent, denoted by $i\simeq j$, 
if we have $v(S\cup\{i\})=v(S\cup\{j\})$ for all $S\subseteq N\backslash\{i,j\}$. 
\end{definition}
The relation $\simeq$ is an equivalence relation and partitions the set of voters $N$ into, say $t$, 
disjoint subsets $N_1,\dots,N_t$ -- the equivalence classes of voters. Roughly speaking, adding voter 
$i$ instead of voter $j$ to any coalition $S$ will have the same or better effect on its decisiveness, 
making $i$ a more desirable addition for the voters comprising $S$. The following three types of voters 
deserve special attention.
\begin{definition}\label{def_dummy}
Given a simple game $v$, a voter $i\in S$ with $v(S)=v(S\cup\{i\})$ for all $S\subseteq N\backslash\{i\}$ 
is called a dummy.
\end{definition}
A dummy has no bearing on the success of a coalition she is a member of, and is, therefore, powerless.
\begin{definition}\label{def_vetoer}
Given a simple game $v$, a voter $i\in N$ such that $i$ is contained in all minimal winning coalitions 
is called a vetoer. 
\end{definition}
Any voter in a minimal winning coalition is critical to the success of the coalitions. This means that a 
voter present in all minimal winning coalitions has the power of a veto.
\begin{definition}\label{def_dictator}
Given a simple game $v$, a voter $i\in N$ such that $\{i\}$ is the unique minimal winning coalition is 
called a dictator.
\end{definition}
Being a dictator is the strongest form of having a veto. A dictator has all the power, rendering all 
other voters dummies.

Let us now recall some well-known facts about representations of weighted games:
\begin{lemma}\label{lemma_representations}
Each weighted game $v$ admits a representation $(q,w_1,\dots,w_n)$ with $w_1,\dots,w_n\ge 0$, $q>0$, and 
\begin{enumerate}
\item[(1)] $\sum_{i=1}^n w_i=1$, $q\in(0,1]$;
\item[(2)] $\sum_{i=1}^n w_i=1$, and $w_i=0$ for all dummies $i\in N$;
\item[(3)] $q\in\mathbb{N}$, $w_i\in \mathbb{N}$;
\item[(4)] $q\in\mathbb{N}$, $w_i\in \mathbb{N}$, $w_i=w_j$ for all $i\simeq j$, and $w_i=0$ for all 
dummies $i\in N$.
\end{enumerate}
\end{lemma}
We call (1) a normalized representation, and (3) an integer representation. Whenever we have $w_i=w_j$ 
for all $i\simeq j$, we say that the representation is \emph{type-revealing}. A representation 
with $w_i=0$ for all dummies $i\in N$ is called \emph{dummy-revealing}. Given a general (integer) 
representation, the problem of verifying that a voter is a dummy is coNP-complete (Theorem 4.4 in 
\citeasnoun{chalk_elk_wool1}). The generating functions offer an efficient way of finding dummy 
players in weighted voting games \cite{bilb_fern_los_lop2}.

\subsection{Coherent power measures}\label{subsec:coherency}
Let $\mathcal{S}_n$ denote the set of simple games on $n$ voters, and $\mathcal{W}_n\subset\mathcal{S}_n$ 
the set of weighted games on $n$ voters.
\begin{definition}\label{def_power_index}
A power index for $\mathcal{C}\in\{\mathcal{S}_n,\mathcal{W}_n\mid n\in\mathbb{N}\}$ is a mapping 
$g:\mathcal{C}\to\mathbb{R}^n$, where $n$ denotes the number of voters in each game of $\mathcal{C}$.
\end{definition}
We define a vector-valued power index by defining its element $g_i$, the voting power of voter $i$. 
A power index should satisfy the following essential properties:
\begin{definition}\label{def_coherency}
Let $g:\mathcal{C}\rightarrow \mathbb{R}^{|N|}=(g_i)_{i\in N}$ be a power index for $\mathcal{C}$. 
We say that
\begin{enumerate}
\item[(1)] $g$ is symmetric if for all $v\in\mathcal{C}$ and any bijection $\tau:N\rightarrow \tau$ 
we have $g_{\tau(i)}(\tau v)=g_i(v)$, where $\tau v(S)=v(\tau(S))$ for all $S\subseteq N$;
\item[(2)] $g$ is positive if $g_i(v)\ge 0$ and $g(v)\neq 0$ for all $v\in\mathcal{C}$;
\item[(3)] $g$ is efficient if $\sum_{i=1}^n g_i(v)=1$ for all $v\in\mathcal{C}$;
\item[(4)] $g$ satisfies the dummy property if for all $v\in\mathcal{C}$ and all dummies $i$ of 
$v$ we have $g_i(v)=0$.
\end{enumerate}
\end{definition}
Any positive power index $g$ can be made efficient by rescaling: $g'_i(v)=g_i(v)/\sum_{i=1}^n g_i(v)$. 
Rescaling turns the Penrose-Bazhaf absolute measure into the Banzhaf index. The Banzhaf index and the 
Shapley-Shubik index have all the above properties.

In addition to the above properties, any reasonable measure of voting power should recognize the 
equivalence classes of voters. To formalize this property, we need the notion of desirability 
introduced in \citeasnoun{isbell1}:
\begin{definition}\label{def_des_relation}
Given a simple $v$, we write $i\succeq j$, if we have $v(S\cup\{i\})\ge v(S\cup\{j\})$ for all 
$S\subseteq N\backslash\{i,j\}$ and say that voter~$i$ is at least as desirable as voter $j$. 
\end{definition}
We can have $i\succeq j$ and $j\succeq i$, if and only if $i\simeq j$. In this case, voters $i$ 
and $j$ are equivalent in the sense of belonging to the same equivalence class. We say $i\succ j$, 
if $i\succeq j$ and $i\not\simeq j$. In an arbitrary simple game, we can have $i\not\succeq j$ and 
$j\not\succeq i$. In this case, the two voters $i,j\in N$ are incomparable. To exclude this 
possibility, a class of games narrower than simple games but still more general than weighted 
voting games has been proposed by \citeasnoun{isbell1} and elaborated in \citeasnoun{taylor_zwick1}.
\begin{definition}\label{def_complete_game}
A simple game $v$ is called complete, if we have $i\succeq j$ or $j\succeq i$ (including both 
possibilities) for all voters $i,j\in N$.
\end{definition}
\citeasnoun{taylor_pacel1} offer a test of completeness. A simple game is complete if it is 
swap robust, or if a one-for-one exchange of players between any two winning coalitions $S$ 
and $T$ leaves at least one of the two coalitions winning. One of the players in the swap 
must belong to $S$ but not $T$, and the other must belong to $T$ but not $S$.

It is important to emphasize that all weighted games are complete, so that the $\succeq$-relation 
induces a complete, or total, ordering of the voters. Given a 
representation $(q,w_1,\dots,w_n)$, $w_i\ge w_j$ implies $i\succeq j$, and $w_i=w_j$ 
implies $i\simeq j$. The implication $i\succ j$ from $w_i>w_j$ is only valid if the given 
representation preserves types formed by the partition of voters according to the equivalence 
relationship.

Being simple, a complete game can be defined by the set of minimal winning coalitions. For 
complete games, however, a still more parsimonious definition based on shift-minimal winning
and shift-maximal losing coalitions is available.
\begin{definition}\label{def_shift_coalitions}
Let $v$ be a complete simple game, where $1\succeq 2\succeq \dots\succeq n$, and $S\subseteq N$ 
be a coalition. A coalition $T\subseteq N$ is a direct left-shift of $S$ whenever there exists a 
voter $i\in S$ with $i-1\notin S$ such that $T=S\backslash\{i\}\cup\{i-1\}$ for $i>1$ or 
$T=S\cup\{n\}$ for $n\notin S$. Similarly, a coalition $T\subseteq N$ is a direct right-shift 
of $S$ whenever there exists a voter $i\in S$ with $i+1\notin S$ such that 
$T=S\backslash\{i\}\cup\{i+1\}$ for $i<n$ or $T=S\backslash\{n\}$ for $n\in S$.

A coalition $T$ is a left-shift of $S$, if it arises as a sequence of direct left-shifts. 
Similarly, it is a right-shift of $S$ if it arises as a sequence of direct right-shifts. 
A winning coalition $S$ such that all right-shifts of $S$ are losing is called shift-minimal 
winning. Similarly, a winning coalition $S$ such that all left-shifts of $S$ are winning is 
called shift-maximal losing. 
\end{definition}
A complete game is uniquely defined by either the set of shift-minimal winning coalitions, 
or the set of shift-maximal losing coalitions. The minimal winning coalitions of the game 
$[51;47,46,5,2]$ discussed in the introduction are $\{\{1, 2\}, \{1, 3\}, \{2, 3\}\}$, while 
the maximal losing coalitions are $\{\{1, 4\}, \{2, 4\}, \{3, 4\}\}$. In this example, all 
minimal and maximal coalitions are also shift-minimal and shift-maximal. This is not the case 
in general. The shift coalitions form subsets of the set of their respective winning and losing 
coalitions. For example, the maximal losing coalitions of the game $[5;3,2,2,1]$ are given by 
$\{\{1,4\}, \{2,3\}, \{2,4\}, \{3,4\}\}$, yet only the former two coalitions are shift-maximal losing.

\begin{definition}\label{def_strong_monotonicity}
A power index $g:\mathcal{C}\rightarrow \mathbb{R}^{|N|}=(g_i)_{i\in N}$ for $\mathcal{C}$ satisfies 
strong monotonicity if we have $g_i(v)>g_j(v)$ for all $v\in\mathcal{C}$ and all voters with $i\succ j$ 
in $v$.
\end{definition}
According to \citeasnoun{freixas_gamba1}, a power index is coherent if it satisfies the four 
properties of Definition~\ref{def_coherency} and is strongly monotonic. Strongly monotonicity 
ensures that the power index recognizes the equivalence classes of voters.

\subsection{Representation-compatibility}\label{subsec:repcompatibility}
In the next section, we introduce two power indices for weighted games that respects the 
proportionality of power and weight. We call such power indices \emph{representation-compatible}.
\begin{definition}\label{def_representation_compatible}
A power index $g:\mathcal{W}_n\rightarrow\mathbb{R}^n$ for weighted games on $n$ voters is 
called represent-\newline ation-compatible if $(g_1(v),\dots,g_n(v))$ is feasible for all 
$v\in\mathcal{W}_n$.
\end{definition}
The existing power measures are not representation-compatible in general. For example, the 
Banzhaf index (BZI) and the Shapley-Shubik index (SSI) are representation-compatible for 
$n\le3$ only. Table~\ref{tab1} compares all the weighted games with up to three voters in 
minimum sum integer representations to the respective power distribution according to the 
two measures.

\renewcommand*{\arraystretch}{1.4}
\begin{longtable}[c]{p{2cm}p{2cm}p{2cm}p{2cm}p{2cm}p{2cm}}
\caption[]{\begin{minipage}[c]{13.5cm} Representation-compatibility of the BZI and the SSI for
$n\le3$.\end{minipage}}\label{tab1}\endhead\hline\hline
Game & BZI & SSI & Game & BZI & SSI\\\hline
$[1;1]$ & $\left[1;1\right]$ & $\left[1;1\right]$ & $[2,1,1,0]$ & $\left[\frac{6}{6};\frac{3}{6},\frac{3}{6},\frac{0}{6}\right]$ & $\left[\frac{6}{6};\frac{3}{6},\frac{3}{6},\frac{0}{6}\right]$\\
$[1;1,0]$ & $\left[\frac{2}{2};\frac{2}{2},\frac{0}{2}\right]$ & $\left[\frac{2}{2};\frac{2}{2},\frac{0}{2}\right]$ & $[1;1,1,1]$ & $\left[\frac{2}{6};\frac{2}{6},\frac{2}{6},\frac{2}{6}\right]$ & $\left[\frac{2}{6};\frac{2}{6},\frac{2}{6},\frac{2}{6}\right]$\\
$[1;1,1]$ & $\left[\frac{1}{2};\frac{1}{2},\frac{1}{2}\right]$ & $\left[\frac{1}{2};\frac{1}{2},\frac{1}{2}\right]$ & $[2;1,1,1]$ & $\left[\frac{4}{6};\frac{2}{6},\frac{2}{6},\frac{2}{6}\right]$ & $\left[\frac{4}{6};\frac{2}{6},\frac{2}{6},\frac{2}{6}\right]$\\
$[2,1,1]$ & $\left[\frac{2}{2};\frac{1}{2},\frac{1}{2}\right]$ & $\left[\frac{2}{2};\frac{1}{2},\frac{1}{2}\right]$ & $[3;1,1,1]$ & $\left[\frac{6}{6};\frac{2}{6},\frac{2}{6},\frac{2}{6}\right]$ & $\left[\frac{6}{6};\frac{2}{6},\frac{2}{6},\frac{2}{6}\right]$\\
$[1;1,0,0]$ & $\left[\frac{6}{6};\frac{6}{6},\frac{0}{6},\frac{0}{6}\right]$ & $\left[\frac{6}{6};\frac{6}{6},\frac{0}{6},\frac{0}{6}\right]$ & $[3;2,1,1]$ & $\left[\frac{4}{5};\frac{3}{5},\frac{1}{5},\frac{1}{5}\right]$ & $\left[\frac{5}{6};\frac{4}{6},\frac{1}{6},\frac{1}{6}\right]$\\
$[1;1,1,0]$ & $\left[\frac{3}{6};\frac{3}{6},\frac{3}{6},\frac{0}{6}\right]$ & $\left[\frac{3}{6};\frac{3}{6},\frac{3}{6},\frac{0}{6}\right]$ & $[2;2,1,1]$ & $\left[\frac{2}{5};\frac{3}{5},\frac{1}{5},\frac{1}{5}\right]$ & $\left[\frac{2}{6};\frac{4}{6},\frac{1}{6},\frac{1}{6}\right]$\\\hline\hline
\end{longtable}

For $n\ge 4$, one can easily find examples in which the Shapley-Shubik power vector is not 
representation-compatible. For example, take the representation $[3;2,1,1,1]$. The corresponding 
Shapley-Shubik power vector is given by $\left(\frac{1}{2},\frac{1}{6},\frac{1}{6},\frac{1}{6}\right)$. 
Since $\{2,3,4\}$ is a winning coalition with weight $\frac{1}{2}$, and $\{1\}$ is a losing coalition 
with weight $\frac{1}{2}$, the Shapley-Shubik power vector cannot be a representation of the game. The 
same counter-example also applies for the Banzhaf index, since in this game the two power vectors 
coincide. \citeasnoun{houy_zwicker1} characterize the set of representations that is compatible with 
the Banzhaf index in a general weighted game.

It is not a coincidence that some power vectors in Table~\ref{tab1} occur several times. This follows 
from duality.
\begin{definition}\label{def_duality}
Let $v:2^N\rightarrow\{0,1\}$ be a simple game and $\mathcal{W}$ its set of winning coalitions, 
$\mathcal{L}$ its set of losing coalitions. By $v^d:2^N\rightarrow\{0,1\}$, with $v^d(S)=1-v(N\backslash S)$ 
for all $S\subseteq N$, we denote its dual game.
\end{definition}

The Shapley-Shubik power vector, as well as the Banzhaf vector, of a simple game $v$ coincides with 
that of its dual game $v^d$.\footnote{See, for example, the discussion in Chapter 6.2 in
 \citeasnoun{felsen_mach1}.} The result follows because $v$ and $v^d$ may also coincide. A 
weighted representation for the dual game can be obtained from a representation of the original 
game:
\begin{lemma}\label{lemma_dual_weights}
Let $v$ be a weighted game with integer representation $(q,w_1,\dots,w_n)$, and let 
$w(S)=\sum_{i\in S}w_i$, then
$$
\left(w(N)-q+1,w_1,\dots,w_n\right)
$$
is an representation of its dual game $v^d$.
\end{lemma}

\section{Representation-compatible power indices}\label{sec:definition}
The power indices studied in this paper use the following notions of feasibility and 
representation-compatibility. The first notion applies to a normalized vector of voting weights, 
whereas the second notion applies to a representation.
\begin{definition}\label{def_feasibility}
Given a weighted game $v$, a vector $(q,w_1,\dots,w_n)$ is a representation of $v$ if 
$v=[q;w_1,\dots,w_n]$. A weight vector $(w_1,\dots,w_n)$ is called feasible for $v$ if 
there exists a quota $q$ such that $(q;w_1,\dots,w_n)$ is a representation of $v$.
\end{definition}
For a normalized vector of weights to be feasible, it must fulfill the linear inequality 
constraints imposed by the set of minimal winning coalitions and the set of maximal losing 
coalitions.
\begin{lemma}\label{lemma_set_of_normalized_feasible_weights}
The set of all normalized weight vectors $w\in\mathbb{R}^n_{\ge 0}$, $\sum_{i=1}^n w_i=1$ 
being feasible for a given weighted game $v$ is given by the intersection
$$
\sum_{i\in S} w_i>\sum_{i\in T} w_i
$$
for all pairs $(S,T)$, where $S$ is a minimal winning and $T$ is a maximal losing coalition 
of $v$.
\end{lemma}
Similarly, for a representation to be valid, or compatible with a given weighted game, it 
must fulfill the linear inequality constraints imposed by the set of minimal winning 
coalitions and the maximal losing coalitions of the game.
\begin{lemma}\label{lemma_set_of_normalized_representations}
The set of all normalized representations $(q,w)\in\mathbb{R}^{n+1}_{\ge 0}$, $q\in(0,1]$, 
$\sum_{i=1}^n w_i=1$ representing a given weighted game $v$ is given by the intersection
$$
\sum_{i\in S} w_i\ge q,\quad\sum_{i\in T} w_i<q
$$
for all minimal winning coalitions $S$ and all maximal losing coalitions $T$.
\end{lemma}
The two sets of linear inequalities define convex polytopes in Euclidean space.

The following lemma shows that we can replace the strict inequalities by the corresponding 
non-strict inequalities, because after the elimination of one weight, the resulting polytopes 
(defined below) are full dimensional. The dimensions are $n$ in the case of the polytope 
defined in Lemma~\ref{lemma_set_of_normalized_representations}, and $n-1$ in the case of the 
polytope defined in (Lemma~\ref{lemma_set_of_normalized_feasible_weights}). We have,
\begin{lemma}\label{lemma_full_dimension}
For each weighted game $v$ there exist positive real numbers $\tilde{q},\tilde{w}_1,
\dots,\tilde{w}_{n-1}$, and a parameter $\alpha>0$, such that 
$$
\left(\tilde{q}+\delta_0,\tilde{w}_1+\delta_1,\dots,\tilde{w}_{n-1}+\delta_{n-1},1-
\sum_{i=1}^{n-1}\left(\tilde{w}_i+\delta_i\right)\right)
$$
is a normalized representation of $v$ for all $\delta_i\in[-\alpha,\alpha]$, $0\le i\le n-1$. 
\end{lemma}
\begin{proof}
Let $(q,w_1,\dots,w_n)$ be an integer representation of $v$. Consequently, the weight of each 
winning coalition is at least $q$, and the weight of each losing coalition is at most $q-1$. 
Since $\big((n+1)q,(n+1)w_1+1,\dots,(n+1)w_n\big)$ is also an integer representation of $v$, we 
additionally assume, without any loss of generality, that $w_i\ge 1$ for all $1\le i\le n$. One 
can easily check that
$$
\left(q-\frac{2}{5}+\tilde{\delta}_0,w_1+\tilde{\delta}_1,\dots,w_n+\tilde{\delta}_n\right)
$$
is a representation of $v$ for all $\tilde{\delta}_i\in\left[-\frac{1}{5n},\frac{1}{5n}\right]$, 
$0\le i\le n$. With $s=\sum_{i=1}^n w_i$, let $\tilde{q}=\frac{1}{s}\cdot\left(q-\frac{2}{5}\right)$ 
and $\tilde{w}_i=\frac{1}{s}\cdot w_i$ for all $1\le i\le n-1$. The choice of a suitable $\alpha$ is 
fiddly. For example, $\alpha=\frac{1}{5ns}$ is too large, whereas $\alpha=\frac{1}{10ns}$ works, but 
the existence is guaranteed by construction.
\end{proof}

To formally define the polytopes, let $\mathcal{W}^m$ be the set of minimal winning coalitions and 
$\mathcal{L}^m$ the set of maximal losing coalitions. The weight polytope is given by
$$
\V(v)=\left\{w\in\mathbb{R}^n_{\ge 0}\mid \sum_{i=1}^n w_i=1,\,w(S)\ge w(T)\quad
\forall S\in\mathcal{W}^m,\,T\in\mathcal{L}^m\right\}.
$$
The representation polytope is given by
$$
\R(v)=\left\{(q,w)\in\mathbb{R}^{n+1}_{\ge 0}\mid \sum_{i=1}^n w_i=1,\,w(S)\ge q\quad
\forall S\in\mathcal{W}^m,\,w(T)\le q\quad\forall T\in\mathcal{L}^m\right\}.
$$

Let us illustrate the computation of average normalized weights on the example discussed in 
the introduction. The weight polytope $\mathsf{V}=\mathsf{V}(v)$ of the game $v=[51;47,46,5,2]$ 
is defined by the following system of inequalities
\begin{eqnarray*}
w_1+w_2 \ge w_1+w_4,\quad w_1+w_2 \ge w_2+w_4,\quad w_1+w_2 \ge w_3+w_4,\\
w_1+w_3 \ge w_1+w_4,\quad w_1+w_3 \ge w_2+w_4,\quad w_1+w_3 \ge w_3+w_4,\\
w_2+w_3 \ge w_1+w_4,\quad w_2+w_3 \ge w_2+w_4,\quad w_2+w_3 \ge w_3+w_4,
\end{eqnarray*}
in addition to $w_1+w_2+w_3+w_4=1$ and $w_i\ge 0$. Eliminating redundant inequalities yields
$$
w_1+w_2\ge w_3+w_4,\,\,\,w_1+w_3\ge w_2+w_4,\,\,\,w_2+w_3\ge w_1+w_4,\,\,\,w_4\ge 0,\,\,\,w_1+w_2+w_3+w_4=1.
$$
The variables $w_1,w_2,w_3$ are symmetric. By assuming a specific ordering of these variables, 
we can decompose the integration domain $\mathsf{V}$ into six parts $\mathsf{P}$, such that the 
resulting six integrals are equal. Moreover, it suffices to compute the average normalized weight 
for voter~$4$, because
$$
\int\limits_{\mathsf{V}}w_1\od \mathsf{V}=\int\limits_{\mathsf{V}}w_2\od \mathsf{V}=
\int\limits_{\mathsf{V}}w_3\od \mathsf{V}.
$$

Let the ordering be $w_1\ge w_2\ge w_3$. Substituting $w_1=1-w_2-w_3-w_4$ yields 
$$
\mathsf{P}=\left\{(w_2,w_3,w_4)\in\mathbb{R}^3\mid w_2\ge w_3\ge w_4\ge 0,\, 
2w_2\ge 1-2w_3,\,2w_2\le 1-w_3-w_4\right\}.
$$
To obtain the integration domain $\mathsf{P}$, note that $\max\{w_4\mid w\in \mathsf{V}\}=\frac{1}{4}$. 
Since $w_1\ge w_2\ge w_3\ge w_4$ and $w_1+w_2+w_3+w_4=1$, the maximum of $w_3$ given $w_4$ is $\frac{1-w_4}{3}$. 
Therefore,
\begin{eqnarray*}
\int\limits_{\mathsf{V}} f(w_4)\od \mathsf{V} &=&
6\int\limits_{\mathsf{P}} f(w_4)\od \mathsf{P} = 6\int\limits_{0}^{\frac{1}{4}} 
\int\limits_{w_4}^{(1-w_4)/3} \int\limits_{\max(w_3,1/2-w_3)}^{(1-w_3-w_4)/2} f(w_4)\od w_2 \od w_3 \od w_4\\
&=& 6\int\limits_{0}^{\frac{1}{4}} \int\limits_{w_4}^{\frac{1}{4}} 
\int\limits_{1/2-w_3}^{(1-w_3-w_4)/2} f(w_4)\od w_2 \od w_3 \od w_4 +\\
&+& 6\int\limits_{0}^{\frac{1}{4}} \int\limits_{\frac{1}{4}}^{(1-w_4)/3} 
\int\limits_{w_3}^{(1-w_3-w_4)/2} f(w_4)\od w_2 \od w_3 \od w_4.
\end{eqnarray*}
Setting $f(w_4)=1$ yields $\frac{1}{96}$ as the volume of $\mathsf{V}$, whereas setting 
$f(w_4)=w_4$ yields $\frac{1}{1536}$. The average normalized weight of voter~$4$ thus equals 
$\frac{1}{16}$. The remaining average weights sum to $\frac{15}{16}$.

Replacing $f(w_4)$ by $w_1,w_2,w_3$ yields $\frac{19}{4608}$, $\frac{5}{2304}$ and $\frac{1}{288}$. 
By the symmetry of $w_1,w_2,w_3$,
$$
\int\limits_{\mathsf{V}} w_1\od \mathsf{V} = \int\limits_{\mathsf{V}} w_2\od \mathsf{V} 
= \int\limits_{\mathsf{V}} w_3\od \mathsf{V} = 6\int\limits_{\mathsf{P}} 
\frac{\overset{=1-w_4}{\overbrace{w_1+w_2+w_3}}}{3}\od \mathsf{P} 
=\frac{1}{3}\cdot \left(\frac{19}{4608}+\frac{1}{288}+\frac{5}{2304}\right) =\frac{5}{1536}.
$$
This yields the following vector of average normalized feasible weights 
$\left(\frac{5}{16},\frac{5}{16},\frac{5}{16},\frac{1}{16}\right)$.

We now consider the computation of the average representation based on the polytope $\mathsf{R}$. 
Since $w_1\ge w_2\ge w_3$, a valid quota $q$ must fulfill $w_1+w_4=1-w_2+w_3\le q\le w_2+w_3$, so that
$$
\mathsf{R}=\mathsf{R}(v)=\left\{(q,w)\mid w\in \mathsf{V},\, 1-w_2+w_3\le q\le w_2+w_3\right\}.
$$
Following the above reasoning, we obtain 
\begin{eqnarray*}
\int\limits_{\mathsf{R}} f(q,w_4)\od \mathsf{R} &=&
6\int\limits_{0}^{\frac{1}{4}} \int\limits_{w_4}^{(1-w_4)/3} 
\int\limits_{\max(w_3,1/2-w_3)}^{(1-w_3-w_4)/2}\int\limits_{1-w_2-w_3}^{w_2+w_3} f(q,w_4) 
\od q\od w_2 \od w_3 \od w_4\\
&=& 6\int\limits_{0}^{\frac{1}{4}} \int\limits_{w_4}^{\frac{1}{4}} 
\int\limits_{1/2-w_3}^{(1-w_3-w_4)/2}\int\limits_{1-w_2-w_3}^{w_2+w_3} f(q,w_4)\od q\od w_2 \od w_3 \od w_4 +\\
&+& 6\int\limits_{0}^{\frac{1}{4}} \int\limits_{\frac{1}{4}}^{(1-w_4)/3} \int\limits_{w_3}^{(1-w_3-w_4)/2}
\int\limits_{1-w_2-w_3}^{w_2+w_3} f(q,w_4)\od q\od w_2 \od w_3 \od w_4.
\end{eqnarray*}
Setting $f(q,w_4)=1$ yields $\frac{1}{1152}$ as the volume of the polytope $\mathsf{R}$. For $f(q,w_4)=w_4$ 
we obtain $\frac{1}{23040}$, so that the average representation of voter~$4$ is given by $\frac{1}{20}$. The 
average representation of the game reads 
$\left(\frac{1}{2};\frac{19}{60},\frac{19}{60},\frac{19}{60},\frac{1}{20}\right)$. The average quota can 
be obtained by setting $f(q,w_4)=q$.

To round off the example, we mention that in small games integrals such as those can be evaluated using 
the software \texttt{LattE} by \citeasnoun{latte}.\footnote{Our programs use \texttt{LattE}. They can be 
downloaded from \texttt{http://serguei.kaniovski.wifo.ac.at/}.} Computing average weights in games with many 
players may require numerical integration based on Monte Carlo methods, such as hit-and-run.

The average normalized weights and the average representation come close to fulfilling the criteria for 
coherent measures of voting power provided in Definition~\ref{def_coherency}. By construction, they are 
symmetric, positive, efficient and strongly monotonic according to Definition~\ref{def_strong_monotonicity}. 
Strong monotonicity in the sense of Isbell's desirability relation in Definition~\ref{def_des_relation} 
follows, because $i\succ j$ implies $w_i>w_j$ in each representation of a given weighted game. However, 
they do not satisfy the dummy property, as this property was not accounted for in the underlying set of 
inequalities. Indeed, in the above example the fourth voter is a dummy, yet her weight in the vector of 
average weights power does not vanish.

To ensure coherency we restrict the polytopes so that all dummies receive the value of zero. The 
dummy-revealing weight polytope is given by
$$
\V^d(v)=\V(v)\cap\left\{w\in\mathbb{R}^n_{\ge 0}\mid w_i=0\quad\forall i\in D\right\}.
$$
The dummy-revealing representation polytope is given by
$$
\R^d(v)=\R(v)\cap\left\{w\in\mathbb{R}^n_{\ge 0}\mid w_i=0\quad\forall i\in D\right\}.
$$
From Lemma~\ref{lemma_full_dimension} we conclude that the $(t-1)$-dimensional volume of 
$\V^d(v)$ and the $t$-dimensional volume of $\R^d(v)$ is non-zero for each weighted game $v$, 
where $1\le t\le n$ denotes the number of non-dummy voters of $v$. We can now use the restricted 
polytopes $\V^d(v)$ and $\R^d(v)$ to define the power indices.
\begin{definition}\label{def_indices}
The average weight index of voter $i$ in a weighted game $v$ is given by
$$
AWI_i(v)=\frac{\int_{\V^d}w_i\od w}{\int_{\V^d}\od w}.
$$
Similarly, the average representation index of voter $i$ in a weighted game $v$ is given by
$$
ARI_i(v)=\frac{\int_{\R^d}w_i\od (q,w)}{\int_{\R^d}\od (q,w)}.
$$
\end{definition}
In the above definition, all integrals are understood as multiple integrals.

It is important to note that dummy-related restrictions are irrelevant for computation. In fact, the 
more dummy voters a game has, the simpler the power computations are. We can safely remove the dummies 
prior to computing the indices. The validity of this procedure follows from the following result, which 
also holds for the Banzhaf index. Given a weighted game $v:2^N\rightarrow\{0,1\}$ with the set of dummy 
voters $D\subset N$, we define the dummy-reduced game $v':2^{N\backslash D}\rightarrow \{0,1\}$ via 
$v'(T)=v(T)$ for all $T\subseteq N\backslash D$. All dummies receive the value of zero in the outcome vector.
\begin{lemma}\label{lemma_dummy}
Given a sequence of power indices $g^n:\mathcal{C}_n\rightarrow\mathbb{R}^n$ for all $n\in\mathbb{N}$, let 
$\tilde{g}^n:\mathcal{C}_n\rightarrow\mathbb{R}^n$ be defined via $\tilde{g}^n_i(v)=g^m_i(v')$ for all 
non-dummies $i$ and by $\tilde{g}^n_j(v)=0$ for all dummies $j$, where $m$ is the number of non-dummies in 
$v$ and $v'$ arises from $v$ by removing the dummies. The power index $\tilde{g}^n$ now satisfies the dummy 
property. 
\end{lemma} 
We call $\tilde{g}^n$ the dummy-revealing version of a given sequence of power indices $g^n$. The above 
lemma shows that the presence of dummies reduces the dimension of the polytopes, thus simplifying computations.

Tables~\ref{tabA2a} and \ref{tabA2a_cont} of Appendix~\ref{app:tables} list power distributions according 
to the AWI and ARI for all weighted games with up to five voters. Power distributions in games with fewer 
than five voters can be obtained from games in which the additional voters are assumed to be dummies. For 
example, the power distribution in the game [3;2,1,1], in which none of the three voters is a dummy, is 
given by the first three coordinates of the power vector for the game [3;2,1,1,0,0], in which the additional 
two voters are dummies. This holds for each of the four power indices.

Ensuring that AWI and ARI preserve the types of voters implied in the equivalence relations of 
Definition~\ref{def_equivalence_classes} requires imposing the following type-revealing restrictions 
on the polytopes.
\begin{eqnarray*}
\V^t(v) &=& \V^d(v)\cap\left\{w\in\mathbb{R}^n_{\ge 0}\mid w_i=w_j\quad\forall i,j\in N\mbox{ s.t. }i\simeq j\right\},\\
\R^t(v) &=& \R^d(v)\cap\left\{w\in\mathbb{R}^n_{\ge 0}\mid w_i=w_j\quad\forall i,j\in N\mbox{ s.t. }i\simeq j\right\}.
\end{eqnarray*}
Lemma~\ref{lemma_full_dimension} implies that $(t-1)$-dimensional volume of $\V^t(v)$ and the $t$-dimensional 
volume of $\R^t(v)$ is non-zero for each weighted game $v$, where $1\le t\le n$ denotes the number of 
equivalence classes of voters of $v$, excluding the dummy voters. Note that the dimension of the polytopes 
for the type-revealing indices are typically smaller than for the AWI and ARI. The case of $t=n$ can be 
handled separately, as in this case all voters are by definition equally powerful.

Tables~\ref{tabA2b} and \ref{tabA2b_cont} of Appendix~\ref{app:tables} list the type-revealing versions 
of AWI and ARI, called AWTI and ARTI, for all weighted games with up to four voters. A formal definition 
of the type-revealing indices is completely analogous to Definition~\ref{def_indices}. The computation 
of AWTI and ARTI follows the same procedures described the example above, except that it uses the 
restricted versions of the polytopes instead of their unrestricted counterparts. A complete example 
of the above calculations is provided in Appendix~\ref{app:example}.

We conclude the presentation of the power indices with a remark on duality (Definition~\ref{def_duality}).
\begin{lemma}\label{lemma_duality}
The average weight index (AWI) and the average representation index (ARI) coincide for the pairs of a 
weighted game $v$ and its dual $v^d$.
\end{lemma}
\begin{proof}
According to Lemma~\ref{lemma_dual_weights}, the integer representations of $v$ and $v^d$ are in bijection. 
Let $(q,w_1,\dots,w_n)$ be a normalized representation of $v$, then $(1-q+\varepsilon,w_1,\dots,w_n)$ is a 
normalized representation of $v^d$ for a sufficiently small $\varepsilon>0$, as $q\in(0,1]$. If we require 
that the weight of each winning coalition in the dual game strictly exceed the quota, then we can choose 
$1-q$ as a quota for the dual game, while also retaining the weights. In view of 
Lemma~\ref{lemma_full_dimension}, this difference between a strict and non-strict inequality can be 
neglected when computing the indices, which proves the lemma.
\end{proof}

\section{The properties of representation-compatible power indices}\label{sec:properties}
The common criteria for choosing an index include the existence of a game-theoretic axiomatization, 
consistency with certain stochastic models of voting or immunity to certain voting paradoxes. 
Table~\ref{tab2} compares the four power indices (AWI, ARI, AWTI, ARTI) to several existing power indices. 
Some of these power indices are well-known, whereas others have only recently been introduced. All indices 
introduced in this paper are coherent measures of power; they satisfy \textbf{Null}, \textbf{Eff}, 
\textbf{Invar} and \textbf{Str.Mon} (Definition~\ref{def_coherency}). Most researchers agree that a power 
index should at least be coherent. Yet two well-known power indices by \citeasnoun{deegan_pack1} and 
\citeasnoun{holler1} violate monotonicity, and are therefore not coherent.

\begin{scriptsize}
\renewcommand*{\arraystretch}{1.0}
\begin{longtable}[c]{lccccccccc}
\caption[]{\begin{minipage}[t]{10cm} Basic properties and immunities to voting paradoxes. \end{minipage}}\label{tab2}\endhead\hline\hline
Index & \textbf{Null} & \textbf{Eff} & \textbf{Invar} & \textbf{Str.Mon} & \textbf{Prop} & \textbf{Type-Rev} & \textbf{Bloc} & \textbf{Don} & \textbf{Bic.Meet} \\\hline
\citeasnoun{shapl_shub1} & $\checkmark$ & $\checkmark$ & $\checkmark$ & $\checkmark$ & & & $\checkmark$ & $\checkmark$ & \\
\citeasnoun{banzh1} & $\checkmark$ & $\checkmark$ & $\checkmark$ & $\checkmark$ & & & & & $\checkmark$ \\
\citeasnoun{johnst1} & $\checkmark$ & $\checkmark$ & $\checkmark$ & $\checkmark$ & & & & & \\
\citeasnoun{deegan_pack1} & $\checkmark$ & $\checkmark$ & $\checkmark$ & & & & & & \\
\citeasnoun{holler1} & $\checkmark$ & $\checkmark$ & $\checkmark$ & & & & & & \\
\citeasnoun{freixas_kan1} & $\checkmark$ & $\checkmark$ & $\checkmark$ & $\checkmark$ & $\checkmark$ & & & & $\checkmark$ \\
AWI & $\checkmark$ & $\checkmark$ & $\checkmark$ & $\checkmark$ & $\checkmark$ & & & & \\
ARI & $\checkmark$ & $\checkmark$ & $\checkmark$ & $\checkmark$ & $\checkmark$ & & & & \\
AWTI & $\checkmark$ & $\checkmark$ & $\checkmark$ & $\checkmark$ & $\checkmark$ & $\checkmark$ & & & \\
ARTI & $\checkmark$ & $\checkmark$ & $\checkmark$ & $\checkmark$ & $\checkmark$ & $\checkmark$ & & & \\\hline\hline
\end{longtable}
\end{scriptsize}

The defining property of the indices studied in this paper is representation-compatibility, which ensures 
proportionality (\textbf{Prop}) between power and weight. The MSR Index introduced 
in \citeasnoun{freixas_kan1} is the only existing power index that has this property. The new indices 
(AWTI and ARWI) are type-revealing (\textbf{Typ.Rev}), a property unique to them. Proportionality between 
power and weight makes representation-compatible indices convenient measures of power.

\subsection{Distributing parliamentary seats: an example}\label{subsec:nationalrat}
To illustrate this convenience, suppose we wish to fill the Austrian parliament (Nationalrat) following the 
general election of 2013. Six parties have attained the electoral threshold of 4 percent required to secure 
a seat in parliament. Their popular votes are listed in the first column of Table~\ref{tab3}. The Austrian 
parliament uses the D'Hondt method to allocate 183 seats among the political parties that passed the 
threshold. The actual seat distribution is given in the second column.

Despite the fact that the D'Hondt method is not based on power computations, the resulting distribution 
of voting power in the parliament resembles the distribution of power implied in the popular vote. This 
occurs because the D'Hondt method tries to achieve proportionality, thus preserving the game representation 
implied in the popular vote for any given voting rule. In our example, the resemblance is complete. For 
example, under plurality voting rule we have the following weighted voting games based on popular votes 
and parliamentary seats, respectively: $[2215538; 1258605, 1125876, 962313, 582657, 268679, 232946]$ and 
$[92; 52, 47, 40, 24, 11, 9]$. The third column of Table~\ref{tab3} shows that these games have identical 
power distributions according to the Shapley-Shubik index (SSI).

Suppose that, instead of using the D'Hondt method, we allocated the parliamentary seats according to the 
distribution of power implied in the popular vote under plurality voting rule. For example, we could 
distribute the seats according to power distributions obtained from the SSI or AWI, with the implied seat 
distributions provided in Table~\ref{tab3}. The distribution of seats according to the SSI index implies a 
different power distribution than the one given in the third column, as power vectors of the games 
$[92; 52, 47, 40, 24, 11, 9]$ and $[92; 67, 49, 49, 6, 6, 6]$ differ. On the contrary, 
$[92; 52, 47, 40, 24, 11, 9]$ and $[92; 63, 44, 44, 11, 11, 11]$ have identical power vectors according 
to the AWI. A seat distribution according to the AWI allows us to easily discern the power distribution 
from the weight distribution, because the AWI power vector is a representation of the game.

\renewcommand*{\arraystretch}{1.0}
\begin{longtable}[c]{lccccccc}
\caption[]{\begin{minipage}[t]{13.5cm} Austrian Nationalrat election, 2013. \end{minipage}}\label{tab3}\endhead\hline\hline
& Popular Votes & Seats & SSI & SSI Seats & AWI & AWI Seats \\\hline
SP{\"O} & 1,258,605 & 52 & 0.367 & 67 & 0.342 & 63 \\
{\"O}VP & 1,125,876 & 47 & 0.267 & 49 & 0.242 & 44 \\
FP{\"O} & 962,313 & 40 & 0.267 & 49 & 0.242 & 44 \\
Green & 582,657 & 24 & 0.033 & 6 & 0.058 & 11 \\
Team Stronach & 268,679 & 11 & 0.033 & 6 & 0.058 & 11 \\
NEOS & 232,946 & 9 & 0.033 & 6 & 0.058 & 11 \\\hline
Quota & 2,215,538 & 92 & & 92 & & 92 \\\hline\hline
\multicolumn{8}{l}{\begin{minipage}[l]{13.5cm} \vspace*{0.25cm} \footnotesize The total number of 
parliamentary seats according to the AWI is 184 not 183. This rounding error can be rectified by
 subtracting one seat from the largest party, as this would
leave the power distribution unchanged according to AWI.\end{minipage}}
\end{longtable} 

The distribution of seats based on the AWI uses the representation provided by the popular votes as a 
template. A different problem is that of designing a weighted voting game with an arbitrary given power 
distribution -- a problem of practical importance for institutional design. Although this inverse problem 
may not have an exact solution, an approximate solution for a representation-compatible power index can easily 
be found using a grid search for a quota that minimizes an objective function, say the sum of squared deviations 
between weights and powers. The desired power distribution becomes the weight distribution in the solution, 
appropriately rescaled should integer-valued voting weights be needed. This stands in contrast to the classical 
power indices, whose inverse problems are significantly more difficult.\footnote{See, for example, the 
fixed-point iteration methods for obtaining the inverse solution for the Banzhaf index in 
\citeasnoun{aziz_pater_leech1}.}

\subsection{A comparison with the classical power indices}\label{subsec:comparison}
The above example shows that the SSI is not representation-compatible in games with more than three voters, 
and neither is the BZI. If a power vector is not representation-compatible, then it must lie outside the 
polytope containing the feasible weights. To get a broad picture on how representation-compatible power 
indices differ from the classical indices by Banzhaf (BZI) and Shapley-Shubik (SSI), for each game we compute 
the Euclidean distance between the six measures, and consider the distribution of the distances for all games 
of a given size.

Figure~1 shows the boxplots of the distances for all games with sizes up to a given $n$. Similarly to Tables in 
Appendix~\ref{app:tables}, the games differ in their partitions in the equivalence sets, and are defined in 
terms of the minimum sum representations. The bar in the middle shows the median distance. The top whisker 
ranges from the 99\% quantile to the 75\% quantile. The bottom whisker ranges from the 25\% quantile to the 
1\% quantile. The box thus covers the range of 25-75\%.

The differences between the classical and representation-compatible indices become apparent as $n$ increases. 
The bottom panels suggest that representation-compatible indices lie closer to each other than the BZI and SSI, 
the former appears to lie closer to representation-compatible indices than the latter. The median distance 
between the AWI and the SSI is slightly larger than the median distance between the BZI and the AWI. This 
may suggest that the BZI is more likely to be representation-compatible than the SSI. But a power index 
is representation-compatible if it lies in the interior of the dummy-revealing polytope $\V^d(v)$, whose 
center of mass is the AWI power vector. Being closer to the center of mass does not imply being closer to 
the boundary of the polytope.

Among the four representation-compatible indices, the AWI and ARI appear to lie closer to each other then their 
type-revealing versions. The distances between the representation-compatible indices decreases with an 
increasing $n$, which is not surprising given that polytopes containing representation-compatible power 
distributions are likely to shrink as $n$ increases.

\begin{figure}[ht]
\centering
\subfigure[$n\leq 4$]{
\includegraphics[scale=0.4]{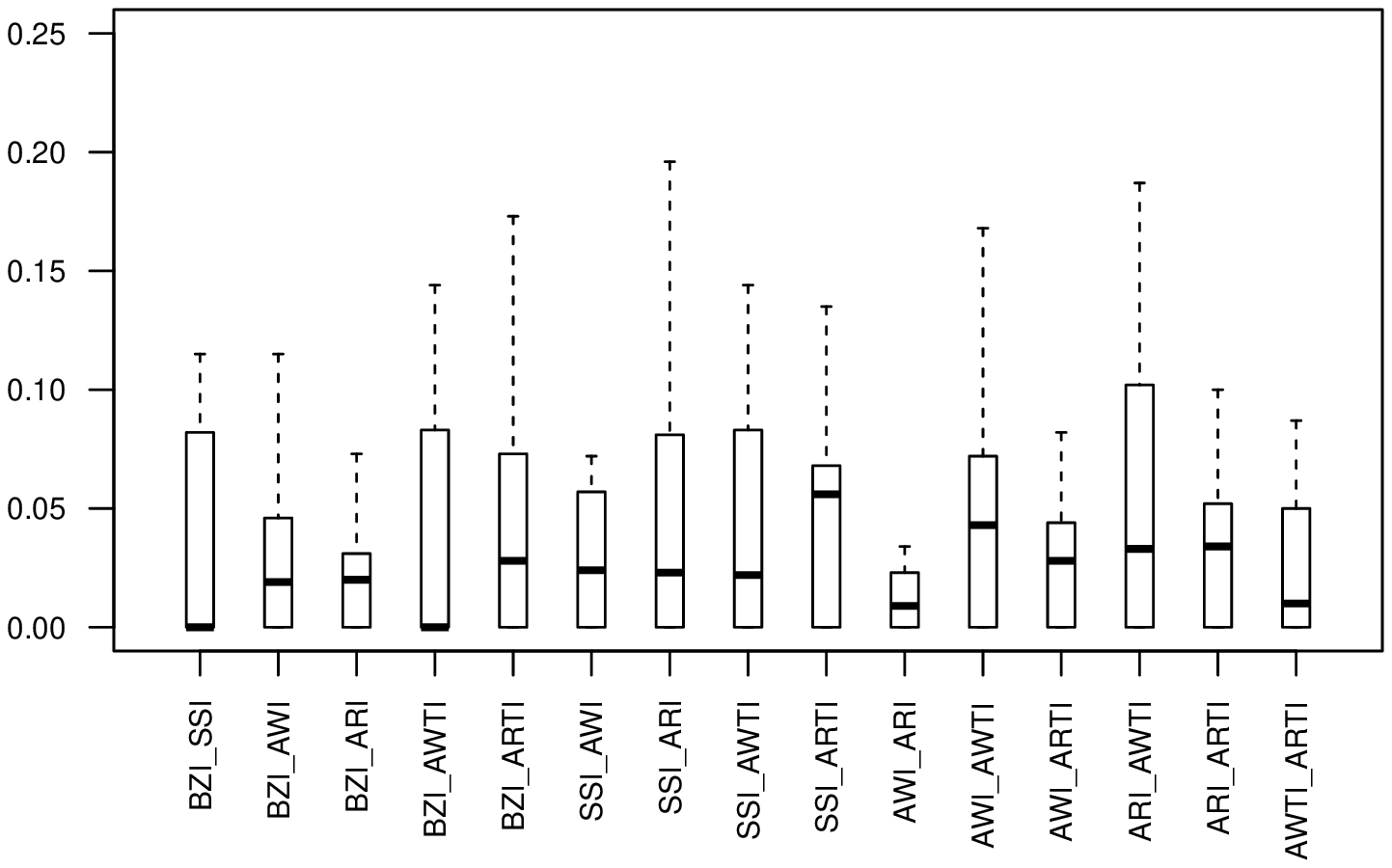}
}
\subfigure[$n\leq 5$]{
\includegraphics[scale=0.4]{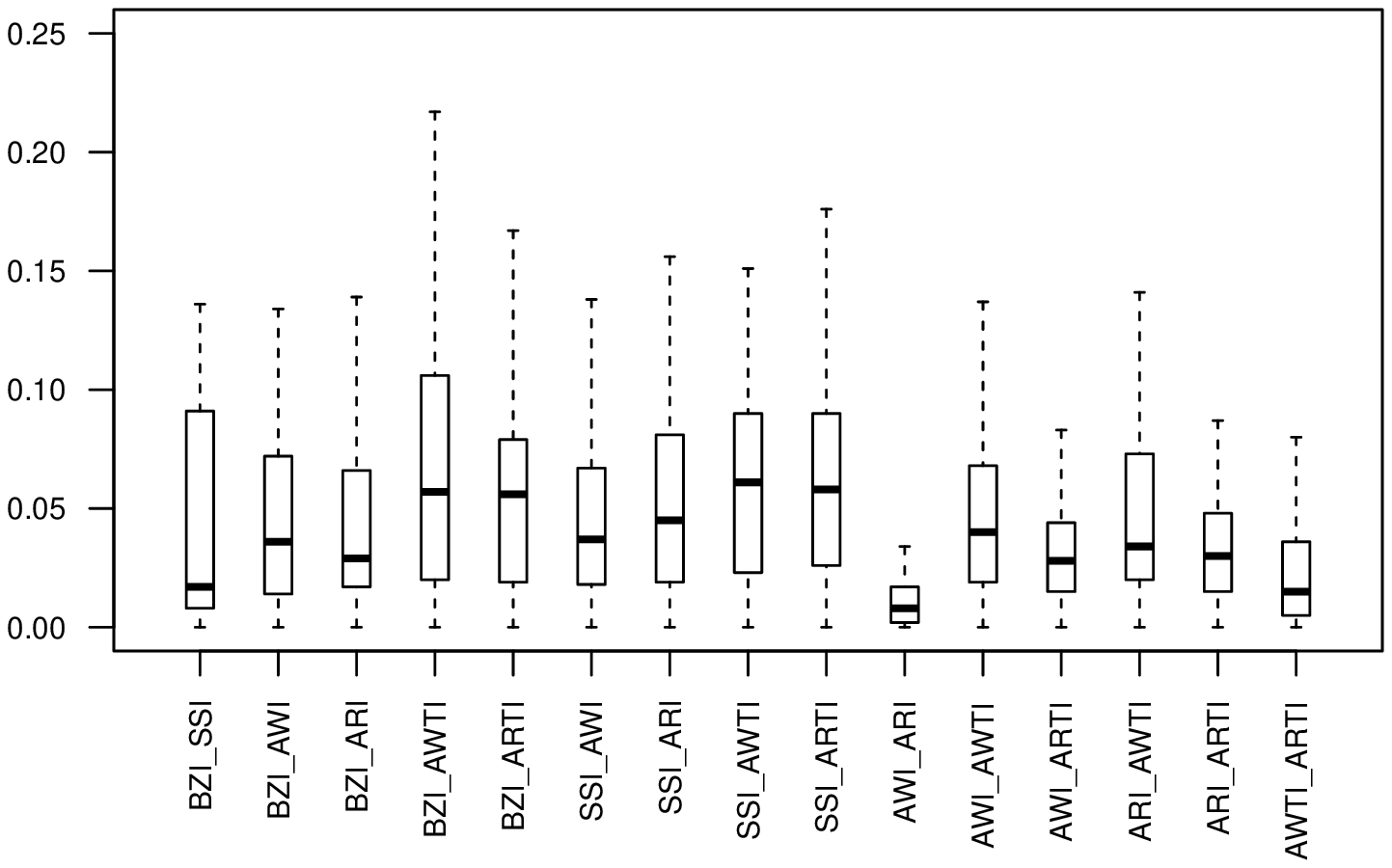}
}
\subfigure[$n\leq 6$]{
\includegraphics[scale=0.4]{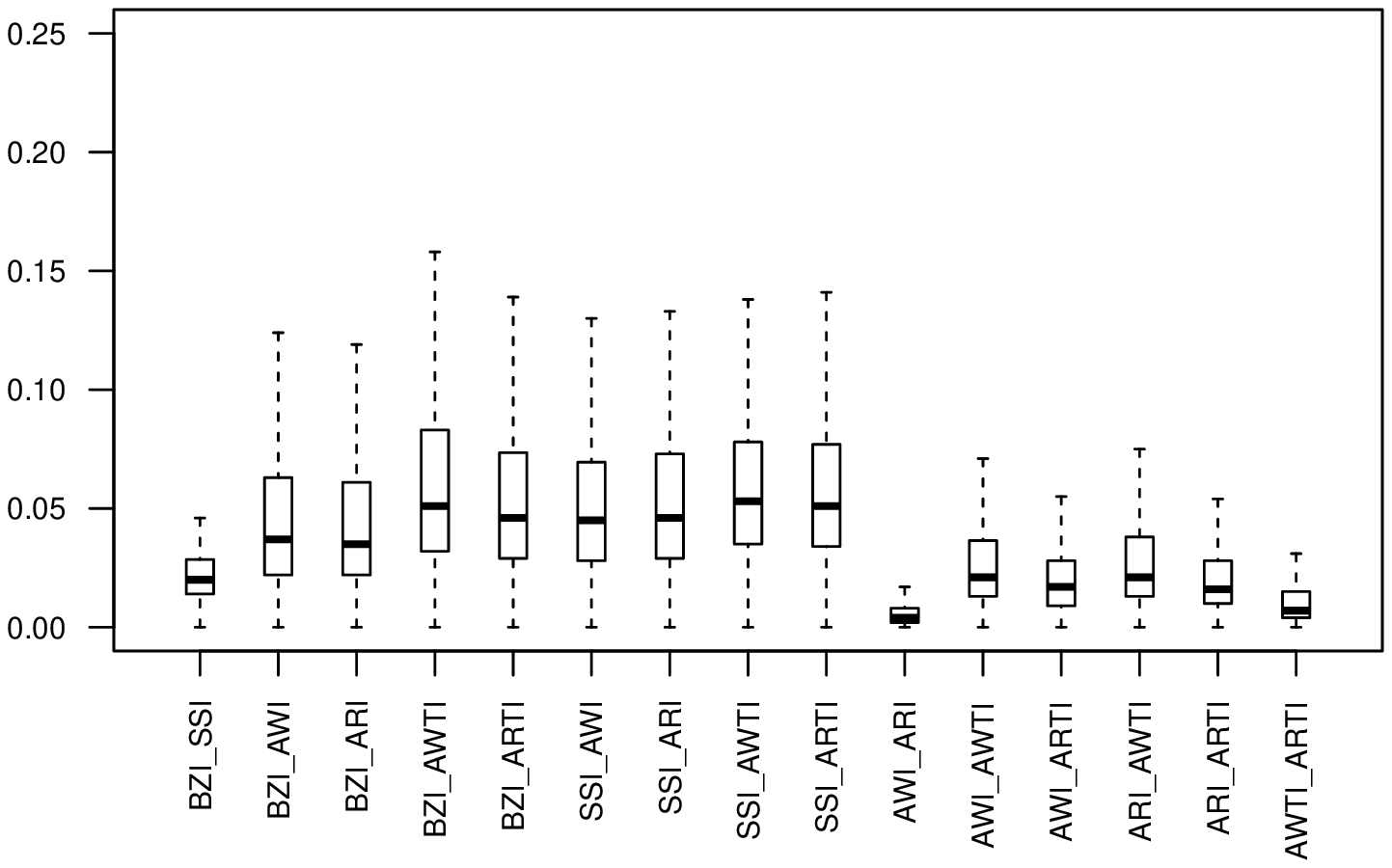}
}
\subfigure[$n\leq 7$]{
\includegraphics[scale=0.4]{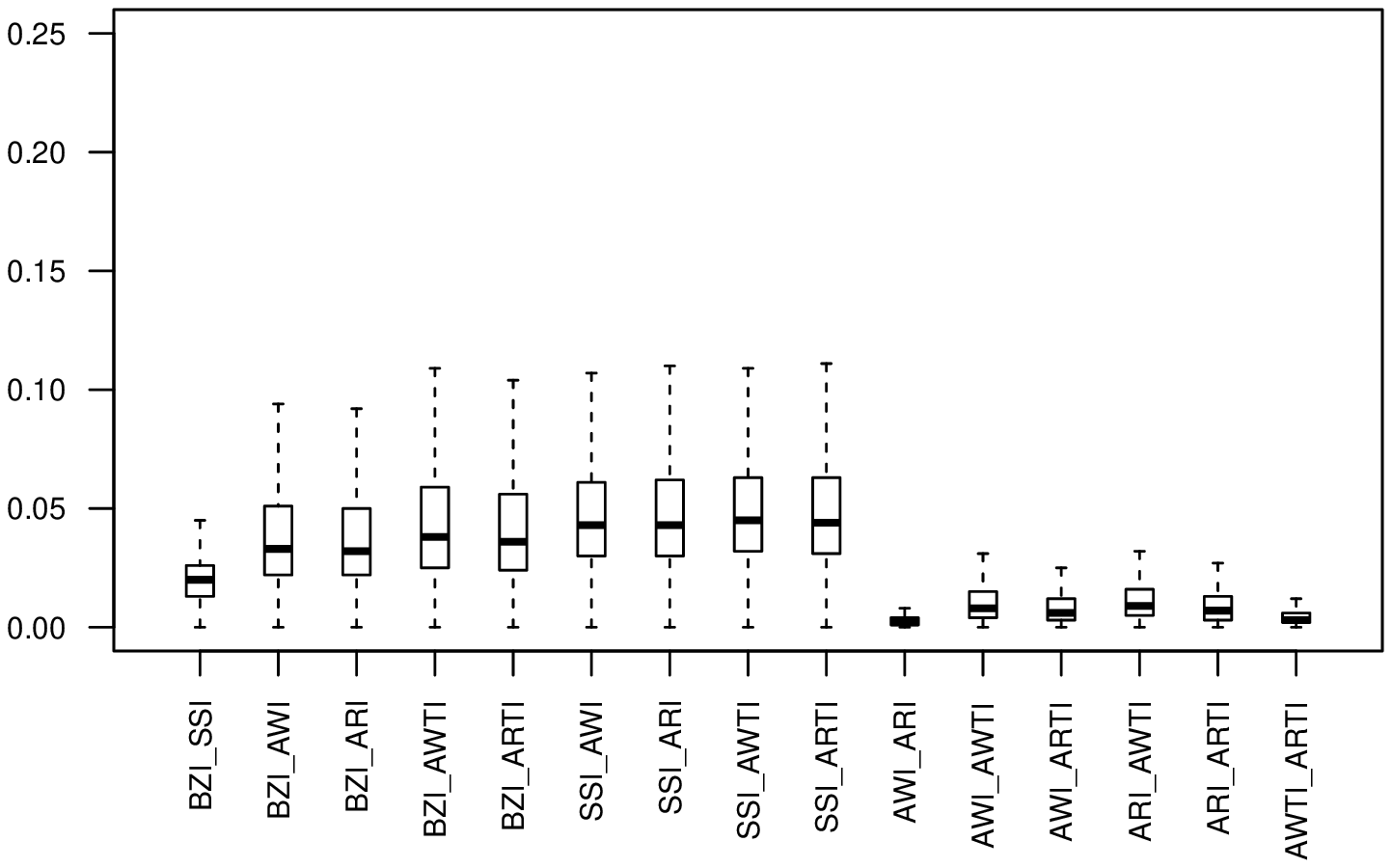}
}
\caption[]{The distribution of Euclidean distances between the indices.}
\end{figure}

\subsection{Immunities to voting paradoxes}\label{subsec:immunities}
\citeasnoun{felsen_mach1} identify three voting paradoxes to which any reasonable measure of power should not 
be liable. These are the bloc (\textbf{Bloc}), donation (\textbf{Don}) and bicameral meet (\textbf{Bic.Meet}) 
paradoxes. In the following, we provide examples showing that all of representation-compatible indices are 
liable to these paradoxes, so they have nothing to recommend in this department. But since none of the existing 
indices have all the required immunities, the question of which index to use cannot be settled based on immunity 
to paradoxes alone.

\noindent\textbf{Bloc paradox:} Respecting the bloc postulate means that if two or more voters form a bloc by 
adding their votes, the power of the bloc should not be lesser than the power of either voter alone. 
Table~\ref{tabfour} provides an example of a game, in which the smallest two voters form a bloc by joining 
their voting weights, and lose power as a result, if only slightly. The BZI and the SSI do not show the 
paradox in this example, although examples are known in which the BZI is liable to the bloc paradox. Also, 
the MSRI of \citeasnoun{freixas_kan1} displays the bloc paradox in this example.

\begin{longtable}[c]{lcccccccc}
\caption[]{\begin{minipage}[t]{8cm} Block paradox in game [37;25,20,17,15,9,6,2,1]. \end{minipage}}\label{tabfour}\endhead\hline\hline
Weight & 25 & 20 & 17 & 15 & 9 & 6 & 2 & 1 \\\hline
BZI & 0.274 & 0.226 & 0.188 & 0.168 & 0.063 & 0.053 & 0.0240 & 0.005 \\
SSI & 0.287 & 0.230 & 0.196 & 0.163 & 0.054 & 0.046 & 0.0202 & 0.004 \\
MSRI & 0.262 & 0.213 & 0.180 & 0.148 & 0.082 & 0.066 & 0.0328 & 0.016 \\
AWI & 0.267 & 0.226 & 0.196 & 0.140 & 0.082 & 0.056 & 0.0283 & 0.006 \\
ARI & 0.266 & 0.224 & 0.194 & 0.140 & 0.082 & 0.057 & 0.0288 & 0.007 \\
AWTI & 0.267 & 0.226 & 0.196 & 0.140 & 0.082 & 0.056 & 0.0283 & 0.006 \\
ARTI & 0.266 & 0.224 & 0.194 & 0.140 & 0.082 & 0.057 & 0.0288 & 0.007 \\\hline
Weight & 25 & 20 & 17 & 15 & 9 & 6 & 3 & 0 \\\hline
BZI & 0.282 & 0.223 & 0.185 & 0.165 & 0.068 & 0.049 & 0.0291 & 0 \\
SSI & 0.293 & 0.226 & 0.193 & 0.160 & 0.060 & 0.043 & 0.0262 & 0 \\
MSRI & 0.273 & 0.212 & 0.182 & 0.152 & 0.091 & 0.061 & 0.0303 & 0 \\
AWI & 0.272 & 0.225 & 0.197 & 0.140 & 0.087 & 0.051 & 0.0281 & 0 \\
ARI & 0.272 & 0.224 & 0.195 & 0.141 & 0.087 & 0.052 & 0.0284 & 0 \\
AWTI & 0.272 & 0.225 & 0.197 & 0.140 & 0.087 & 0.051 & 0.0281 & 0 \\
ARTI & 0.272 & 0.224 & 0.195 & 0.141 & 0.087 & 0.052 & 0.0284 & 0 \\\hline\hline
\end{longtable}

\noindent\textbf{Donation paradox:} Respecting donation means that if one voter gives some of her votes 
to another, the power of the donor should not increase as a result. \citeasnoun{felsen_mach1} provide 
examples in which the Banzhaf and Johnston indices show both bloc and donation paradoxes. 
\citeasnoun{freixas_molin3} study the frequency of the occurrence of the donation paradox in 
weighted games with a small number of players, providing examples for the Banzhaf and Johnston indices. 
The Shapley-Shubik index is immune to both the bloc and donation paradoxes. \citeasnoun{freixas_kan1} 
provide an example, which also shows that the MSR index is liable to the donation paradox.

\begin{longtable}[c]{lcccccc}
\caption[]{\begin{minipage}[t]{7cm} Donation paradox in game [13;9,4,3,2,1]. \end{minipage}}\label{tabfive}\endhead\hline\hline
Weight & 9 & 4 & 3 & 2 & 1 \\\hline
BZI & 0.524 & 0.238 & 0.143 & 0.048 & 0.048 \\
SSI & 0.617 & 0.200 & 0.117 & 0.033 & 0.033 \\
MSRI & 0.417 & 0.250 & 0.167 & 0.083 & 0.083 \\
AWI & 0.518 & 0.247 & 0.138 & 0.048 & 0.048 \\
ARI & 0.501 & 0.247 & 0.143 & 0.054 & 0.054 \\
AWTI & 0.548 & 0.258 & 0.123 & 0.035 & 0.035 \\
ARTI & 0.522 & 0.257 & 0.132 & 0.045 & 0.045 \\\hline
Weight & 8 & 5 & 3 & 2 & 1 \\\hline
BZI & 0.500 & 0.300 & 0.100 & 0.100 & 0 \\
SSI & 0.583 & 0.250 & 0.083 & 0.083 & 0 \\
MSRI & 0.429 & 0.286 & 0.143 & 0.143 & 0 \\
AWI & 0.535 & 0.270 & 0.098 & 0.098 & 0 \\
ARI & 0.513 & 0.273 & 0.107 & 0.107 & 0 \\
AWTI & 0.602 & 0.249 & 0.075 & 0.075 & 0 \\
ARTI & 0.558 & 0.258 & 0.092 & 0.092 & 0 \\\hline\hline
\end{longtable}

The example in Table~\ref{tabfive} shows that representation-compatible indices are liable to the 
donation paradox. In the game $[13;9,4,3,2,1]$, the largest voter gains power by donating one vote to 
the second largest voter according to the Shapley-Shubik index, but gains power according to all other 
indices. In example, the BZI and the MSRI shows the donation paradox. The Shapley-Shubik index is immune to 
both the bloc and donation paradoxes.

\noindent\textbf{The bicameral meet:} An index of power respects bicameral meet if the ratio of powers of 
any two voters belonging to the same assembly prior to a merge with a different assembly is preserved in 
the joint assembly. This property is useful when measuring the voting power of shareholders, because the 
relative powers of shareholders comprising a minority voting assembly with a total voting weight, represented 
by their joint holdings, carries over to the grand voting assembly, represented by the total worth of the 
company.

The bicameral meet of two simple voting games $(N_1, \mathcal W_1)$ and $(N_2, \mathcal W_2)$ is a simple 
voting game $(N, \mathcal W)$, with an assembly $N = N_1 \cup N_2$, and a set of winning coalitions 
$\mathcal W = \{S \subseteq N \; : \; S = S_1 \cup S_2, S_1 \in \mathcal W_1, S_2 \in \mathcal W_2 \}$. 
The two assemblies have no voters in common, so $N_1\cap N_2=\emptyset$. The bicameral meet postulate 
requires that if $i$ and $j$ are non-null voters in a game $(N_1, \mathcal W_1)$, then the ratio of 
power of voter $i$ to the power of voter $j$ in the joint game $(N, \mathcal W)$ should be equal to 
the ratio of their powers in the original game $(N_1, \mathcal W_1)$.

\begin{longtable}[c]{lccccc}
\caption[]{\begin{minipage}[t]{11cm} Added blocker paradox in game $[3;2,1,1]\cup[5;5]=[8;2,1,1,5]$. \end{minipage}}\label{tab6}\endhead\hline\hline
Weight & 2 & 1 & 1 & & Voter 1 / Voter 2 \\\hline
BZI & 0.600 & 0.200 & 0.200 & & 3 \\
SSI & 0.667 & 0.167 & 0.167 & & 4 \\
MSRI & 0.500 & 0.250 & 0.250 & & 2 \\
AWI & 0.611 & 0.194 & 0.194 & & 3.143 \\
ARI & 0.583 & 0.208 & 0.208 & & 2.8 \\
AWTI & 0.667 & 0.167 & 0.167 & & 4 \\
ARTI & 0.611 & 0.194 & 0.194 & & 3.143 \\\hline
Weight & 2 & 1 & 1 & 5 & Voter 1 / Voter 2 \\\hline
BZI & 0.375 & 0.125 & 0.125 & 0.375 & 3 \\
SSI & 0.417 & 0.083 & 0.083 & 0.417 & 5 \\
MSRI & 0.333 & 0.167 & 0.167 & 0.333 & 2 \\
AWI & 0.396 & 0.104 & 0.104 & 0.396 & 3.8 \\
ARI & 0.383 & 0.117 & 0.117 & 0.383 & 3.286 \\
AWTI & 0.375 & 0.125 & 0.125 & 0.375 & 3 \\
ARTI & 0.361 & 0.139 & 0.139 & 0.361 & 2.6 \\\hline\hline
\end{longtable}

\citeasnoun{freixas_kan1} prove that a bicameral meet of two complete games is complete if at least one 
of the two constituent games has only one minimum winning coalition. A special case of the bicameral 
meet postulate is the added blocker postulate, which says that adding a vetoer (Definition~\ref{def_vetoer}) 
to a weighted game should not change the ratio of powers of any two incumbent voters. If an index is 
liable to an added blocker paradox, it is also liable to the bicameral meet postulate, see 
\citeasnoun{felsen_mach1} (p. 270). In Table~\ref{tab7}, we use their example to shows that neither 
of representation-compatible indices satisfies the bicameral meet postulate. The example involves 
adding a blocker with a weight of 5 to the game $[3;2,1,1]$, and adjusting the quote in such a way 
that the set of minimal winning coalitions of the joint games equals the union of the sets of minimal 
winning coalitions in each game. This amounts to joining the games $[3;2,1,1]$ and $[5;5]$. Note that 
the blocker is a dictator is the added game. Since the second game has a single coalition, which is 
trivially minimal winning, the joint game is complete. The bicameral meet postulate does not hold for 
representation-compatible indices, because adding the blocker changes the power ratios of the players. 
This postulate is satisfied by the BZI and the MSRI.

There are many lesser paradoxes and other properties that may distinguish between different indices. 
One useful property is neutrality in symmetric voting games. In a symmetric weighted voting game, each 
player commands an equal number of votes. For a power measure to respect neutrality, the power of a 
voting bloc must equal the sum of individual powers of its members, so that satisfying the bloc postulate 
does not carry strategic implications. The power vectors for games $[3;1,1,1,1]$ and $[3;2,1,1]$ in 
Tables of Appendix~\ref{app:tables} clearly show that representation-compatible indices do not respect 
neutrality. This property is satisfied by the MSRI.

\section{Integral weights and type preservation}\label{sec:integral_weights}
A normalization of voting weights is unreasonable if they represent the number of shares of a corporation 
or the number of members of a political party. In these cases, we require the weights to be integers. 
This observation has led to the development of a power index based on the minimum sum integer representations, 
called the MSR index \cite{freixas_kan1}.

Let us return to the weighted game $v=[2;1,1,1]$. We already mentioned that there exist $1176$ feasible 
integer weight vectors with the total weight of 100. The average of all these vectors equals 
$\left(\frac{100}{3},\frac{100}{3},\frac{100}{3}\right)$, yielding $\left(\frac{1}{3},\frac{1}{3},
\frac{1}{3}\right)$ as the average weight distribution, which is not surprising, given the inherent 
symmetry of the game.

Things get more interesting if we consider the game $v=[3;2,1,1]$. Table~\ref{tab7} lists the number 
of feasible integer weight vectors for an increasing total weight, as well as the average weight 
distributions. The distribution appears to converge to $\left(\frac{11}{18},\frac{7}{36},
\frac{7}{36}\right)$, which equals the AWI for this game. This convergence can be rigorously 
established by numerically approximating the integrals in Definition~\ref{def_indices} over a 
successively finer equally spaced grid inside the polytope. A similar result holds if an 
integer-valued quota is taken into account, in which case we obtain the ARI in the limit. 
The dummy-revealing and type-revealing property is also preserved in the limit.

To obtain power indices based on integer representations, we can minimize the sum of weights 
instead of taking it to infinity. Unfortunately, the minimum sum representations are not unique 
for $n\ge 8$ \cite{kurz1}. Nevertheless, one can take a convex combination of all such minimum sum 
representations, which yields, after a normalization, the MSR index recently introduced in 
\citeasnoun{freixas_kan1}.\footnote{Minimizing $q+\sum_{i=1}^n$ instead of $\sum_{i=1}^n$ makes 
no difference, so there is no need to distinguish the two cost functions.} One motivation for 
minimizing weights is to minimizing the cost of political representation by minimizing the 
number of representatives. Another motivation for minimum sum representation and the MSR 
index is given in \citeasnoun{ansolab_ea1}, who argue that many observations on the formation 
of coalition governments are more consistent with minimal integer-voting weights, rather than 
power distributions implied by the classical power indices.

The minimum sum integer representations are dummy-revealing, but not type-revealing. Take the 
two representations $(12,7,6,6,4,4,4,3,2)$ and $(12,7,6,6,4,4,4,2,3)$, which are minimum sum 
representations of the same game. Indeed, there exists no integer representation of this game 
with a weight sum smaller than 36. Both representations are not type-revealing since the equivalence 
classes of voters are given by $\{1\}$, $\{2,3\}$, $\{4,5,6\}$, and $\{7,8\}$, while $w_7\neq w_8$. 
Implementing one of the two representations may cause some confusion, as one might erroneously think 
that one of the two voters is more powerful than the other. The unique type-revealing minimum sum 
representation for this example is given by $[14;8,7,7,5,5,5,3,3]$. It has a total weight of 43 
instead of 36. Unfortunately, even the type-revealing minimum sum representation can be non-unique 
if at least $n=9$ voters are involved, see \citeasnoun{kurz1}. At the very least, one can define a 
minimum sum representation type-revealing index of a weighted game as a normalization of the convex 
combination of the corresponding set of type-revealing minimum sum representations. Some uniqueness 
results for both the minimum sum representation and the type-revealing minimum sum representation 
exist for special classes of weighted games. \citeasnoun{freixas_kurz1} proved that weighted games 
with up to two equivalence classes of voters admit a unique minimum sum representation. For minimum 
sum representations, the authors give non-unique examples for four equivalence classes of voters and 
conjecture a uniqueness result for three equivalence classes of voters.

For the algorithmic aspects of computing minimum sum representations and the MSR index, see 
\citeasnoun{kurz1}. The gist of this research is that the minimum sum representation can often be 
computed by solving a small sequence of linear programs. The computation of type-revealing minimum sum 
representations requires only minor adjustments.

\begin{longtable}[c]{p{3cm}p{6cm}p{5cm}}
\caption[]{\begin{minipage}[c]{18cm} Convergence of feasible integer weights for $[3;2,1,1]$.\end{minipage}}\label{tab7}\endhead\hline\hline
Total weight & Number of integer representations & Weight distribution\\\hline
100 & 1601 & (0.608832,0.195584,0.195584)\\
1000 & 166001 & (0.610888,0.194556,0.194556)\\
10000 & 16660001 & (0.611089,0.194456,0.194456)\\
100000 & 1666600001 & (0.611109,0.194446,0.194446)\\\hline\hline
\end{longtable}

\section{Concluding remarks}\label{sec:summary}
The average representations of a weighted voting game can be used to obtain four representation-compatible 
indices of voting power for this type of voting game. The average representations are computed from weight 
and representation polytopes defined by the set of winning and losing coalitions of the game. The weight 
polytope is based on normalized voting weights, whereas the representation polytope also includes the quota.

These average representations come remarkably close to fulfilling the standard criteria for a coherent 
measure of voting power. They are symmetric, positive, efficient and strongly monotonic. But they do not 
respect the dummy property that assigns zero power to powerless players. This shortcoming is easily 
rectified by further restricting the polytopes. The resulting restricted average representations respect 
the dummy property and are coherent measures of power.

The above modification suggests that we can endow the indices with qualities by tailoring the polytope. 
Restrictions based on the equivalence classes of voters defined by the Isbell desirability relation lead 
to another pair of power indices, which ascribe equal power to all members of an equivalence class. These 
indices are strictly monotonic in voting weight.

The defining property of the indices is representation-compatibility, which ensures proportionality between 
power and weight. By redistributing weights among the voters, we can redesign any given weighted voting game 
in such a way that the distribution of voting weight will also be the distribution of voting power. This 
allows us to read power directly from the weights, a convenient property that recommends 
representation-compatible indices as optimal representations for weighted voting games, or optimal designs 
for voting institutions. The obvious disadvantage is the computational intensity of integrating monomials 
on highly-dimensional polytopes. Obtaining power distributions in weighted games with many players may 
require numerical integration based on random sampling.

Reflecting on the place representation-compatible indices may take among the existing measures of power, 
we believe that proportionality makes them ideal measures of power for voting institutions, in which the 
votes are distributed to the voter based on their contribution to a fixed purse. In this setting, voting 
power reflects the extent of a voter's control of the distribution of a fixed purse -- the ultimate outcome 
of voting, measured by that voter's expected share in the purse.\footnote{\citeasnoun{felsen_mach3} refer 
to this notion of power as P-power.} If a voter's expected share of spoils coincides with the voter's 
contribution to the fixed purse, an equilibrium emerges in which voters will not wish to redistribute votes. 
This leads to a stable institutional design of vote-for-money institutions, such as a corporation.

\begin{footnotesize}
\bibliographystyle{kluwer}

\end{footnotesize}

\newpage

\begin{appendix}

\section{Appendix}
\subsection{Example [3;2,1,1]}\label{app:example}

The sets of minimal winning and maximal losing coalitions for the game $[3;2,1,1]$ are, respectively, $\{\{1,2\},\{1,3\},\{1,3\}\}$ and $\{\{1\},\{2,3\}\}$. Since voters 2 and 3 are equivalent, there are two equivalence classes in this game. There are no dummies.

Using Lemma~\ref{lemma_set_of_normalized_feasible_weights}, we obtain the following constraints:
\begin{eqnarray*}
w_1+w_2>w_1 & \Longleftrightarrow & w_2>0;\\
w_1+w_3>w_1 & \Longleftrightarrow & w_3>0;\\
w_1+w_2>w_2+w_3 & \Longleftrightarrow & w_1>w_3;\\
w_1+w_3>w_2+w_3 & \Longleftrightarrow & w_1>w_2.
\end{eqnarray*}
In addition, $w_1,w_2,w_3\ge 0$ and $w_1+w_2+w_3=1$. Eliminating $w_3$ and removing the redundant constraints yields the following inequalities: $w_2>0$, $w_2<1-w_1$, $w_2>1-2w_1$, $w_2<w_1$. Since $1-2w_1<w_1$ and $1-w_1>0$, we have $w_1\in\big(\frac{1}{3},1\big)$. For $w_1\in\big(\frac{1}{3},\frac{1}{2}\big)$, we have $w_2\in \big(1-2w_1,w_1\big)$. For $w_1\in\big[\frac{1}{2},1\big)$, we have $w_2\in \big(0,1-w_1\big)$. The polytope is thus given by
$$
\V^d(v)=\left\{(w_1,w_2)\in \mathbb{R}^2_{\ge 0}\mid w_2\ge 0, w_2\le 1-w_1, w_2\ge 1-2w_1, w_2\le w_1\right\}.
$$
Since there are no dummies in this game, the dummy-revealing polytope $\V^d(v)$ coincides with its non-revealing counterpart $\V(v)$.

For voter~$1$, we have
$$
\iint\limits_{\V^d} w_1\od w_1\od w_2=\int_{\frac{1}{3}}^{\frac{1}{2}}\od w_2\int_{1-2w_2}^{w_2} w_1\od w_1+\int_{\frac{1}{2}}^{1}\od w_2\int_{0}^{1-w_2} w_1\od w_1=\frac{1}{54}+\frac{1}{12}=\frac{11}{108}.
$$
For voter~$2$, we obtain
$$
\iint\limits_{\V^d} w_2\od w_1\od w_2=\int_{\frac{1}{3}}^{\frac{1}{2}}\od w_1\int_{1-2w_1}^{w_1} w_2\od w_2+\int_{\frac{1}{2}}^{1}\od w_1\int_{0}^{1-w_1} w_2\od w_2=\frac{1}{48}+\frac{5}{432}=\frac{7}{216}.
$$
The total volume of the polytope is given by
$$
\iint\limits_{\V^d}\od w_1\od w_2=\int_{\frac{1}{3}}^{\frac{1}{2}}\od w_1\int_{1-2w_1}^{w_1}\od w_2+\int_{\frac{1}{2}}^{1}\od w_1\int_{0}^{1-w_1}\od w_2=\frac{1}{8}+\frac{1}{24}=\frac{1}{6},
$$
This yields the following vector of average (normalized) feasible weights $\left(\frac{11}{18},\frac{7}{36},\frac{7}{36}\right)$.

The polytope for the average representation defined by Lemma~\ref{lemma_set_of_normalized_representations} is given by
$$
\R^d(v)=\left\{(q,w_1,w_2)\in\mathbb{R}^3_{\ge 0}\mid w_1+w_2\ge q,w_1\le q,1-w_1\le q,1-w_2\ge q\right\}.
$$
We have,
\begin{eqnarray*}
\iint\limits_{\R^d} w_1\od w_1\od w_2\od q
&=& \int_{\frac{1}{2}}^{\frac{2}{3}}\od q\int_{1-q}^{q}w_1\od w_1\int_{q-w_1}^{1-q}\od w_2+
\int_{\frac{2}{3}}^{1}\od q\int_{2q-1}^{q}w_1\od w_1\int_{q-w_1}^{1-q}\od w_2\\
&=& \frac{31}{7776}+\frac{1}{243}=\frac{7}{864}\\
\iint\limits_{\R^d} w_2\od w_1\od w_2\od q
&=& \int_{\frac{1}{2}}^{\frac{2}{3}}\od q\int_{1-q}^{q}\od w_1\int_{q-w_1}^{1-q}w_2\od w_2+
\int_{\frac{2}{3}}^{1}\od q\int_{2q-1}^{q}\od w_1\int_{q-w_1}^{1-q}w_2\od w_2\\
&=& \frac{29}{15552}+\frac{1}{972}=\frac{5}{1728},\\
\iint\limits_{\R^d}\od w_1\od w_2\od q
&=& \int_{\frac{1}{2}}^{\frac{2}{3}}\od q\int_{1-q}^{q}\od w_1\int_{q-w_1}^{1-q}\od w_2+
\int_{\frac{2}{3}}^{1}\od q\int_{2q-1}^{q}\od w_1\int_{q-w_1}^{1-q}\od w_2\\
&=& \frac{5}{648}+\frac{1}{162}=\frac{1}{72},
\end{eqnarray*}
so that the average representation is given by $\left(\frac{7}{12},\frac{5}{24},\frac{5}{24}\right)$.

The type-revealing power indices require all voters belonging to the same equivalence class to be equally powerful. This assumption is likely to reduce the dimension of the problem, as the number of equivalence classes is typically smaller than the number of voters. We now move from individual voting weights to weights aggregated by equivalence classes, as if voters belonging to the same class form a voting bloc with weight being equal to the sum of weights of its members.

The game has two classes: class A comprises voter 1, whereas voters 2 and 3 form class B. Let $w_a$ be the voting weight of class A, which equals the weight of the first voter $w_a=w_1$. The AWTI polytope degenerates to an interval
$$
\V^t(v)=\left\{w_a\in \mathbb{R}_{\ge 0}\mid 3w_a\ge 1, w_a\le 1\right\}.
$$
We have,
$$
\int\limits_{\V^t} w_a\od w_a=\int_{\frac{1}{3}}^1 w_a\od w_a=\frac{4}{9}\quad\mbox{and}\quad\int\limits_{\V^t} \od w_a=\int_{\frac{1}{3}}^1 \od w_a=\frac{2}{3}.
$$
The voting power of class A according to AWTI equals $\frac{2}{3}$, which is the power of the first voter. The power of class B equals $\frac{1}{3}$. Since all voters comprising a class share its power equally, the AWTI power vector for the voters reads $\left(\frac{2}{3},\frac{1}{6},\frac{1}{6}\right)$.

We now turn to the final index. The ARTI polytope is given by
$$
\R^t(v)=\left\{(q,w_a)\in \mathbb{R}_{\ge 0}^2\mid 3w_a\ge 1, 2q\le 1+w_a, q\ge 1-w_a, q\ge w_a\right\},
$$
where $w_a$ is the voting weight of class A. We have,

\begin{eqnarray*}
\iint\limits_{\R^t} w_a\od w_a\od q
&=& \int_{\frac{1}{3}}^{\frac{1}{2}}w_a\od w_a\int_{1-w_a}^{\frac{1+w_a}{2}}\od q+\int_{\frac{1}{2}}^{1}w_a\od w_a\int_{w_a}^{\frac{1+w_a}{2}}\od q=\frac{1}{108}+\frac{1}{24}=\frac{11}{216},\\
\iint\limits_{\R^t} \od w_a\od q
&=& \int_{\frac{1}{3}}^{\frac{1}{2}}\od w_a\int_{1-w_a}^{\frac{1+w_a}{2}}\od q+\int_{\frac{1}{2}}^{1}\od w_a\int_{w_a}^{\frac{1+w_a}{2}}\od q=\frac{1}{48}+\frac{1}{16}=\frac{1}{12}.
\end{eqnarray*}

The power distribution according to the ARTI is $\left(\frac{11}{18},\frac{7}{36},\frac{7}{36}\right)$.

\subsection{Average representation and type-preserving indices}\label{app:tables}

\renewcommand{\thetable}{A.\arabic{table}}
\setcounter{table}{0}

\begin{landscape}

\begin{footnotesize}
\begin{longtable}[c]{cccccc}
\caption[]{\begin{minipage}[t]{8cm} Average representation indices for $n\leq 5$.\end{minipage}}\label{tabA2a}\endhead\hline\hline
Game & AWI & ARI & Game & AWI & ARI\\\hline
$[1;1,0,0,0,0]$ & $(1.000,0.000,0.000,0.000,0.000)$ & $(1.000,0.000,0.000,0.000,0.000)$ & $[5;3,1,1,1,1]$ & $(0.502,0.125,0.125,0.125,0.125)$ & $(0.489,0.128,0.128,0.128,0.128)$\\
$[1;1,1,0,0,0]$ & $(0.500,0.500,0.000,0.000,0.000)$ & $(0.500,0.500,0.000,0.000,0.000)$ & $[5;3,2,2,2,1]$ & $(0.300,0.198,0.198,0.198,0.104)$ & $(0.300,0.199,0.199,0.199,0.103)$\\
$[1;1,1,1,0,0]$ & $(0.333,0.333,0.333,0.000,0.000)$ & $(0.333,0.333,0.333,0.000,0.000)$ & $[5;4,1,1,1,1]$ & $(0.586,0.104,0.104,0.104,0.104)$ & $(0.571,0.107,0.107,0.107,0.107)$\\
$[1;1,1,1,1,0]$ & $(0.250,0.250,0.250,0.250,0.000)$ & $(0.250,0.250,0.250,0.250,0.000)$ & $[5;4,2,2,1,1]$ & $(0.424,0.209,0.209,0.079,0.079)$ & $(0.420,0.207,0.207,0.083,0.083)$\\
$[1;1,1,1,1,1]$ & $(0.200,0.200,0.200,0.200,0.200)$ & $(0.200,0.200,0.200,0.200,0.200)$ & $[5;4,3,2,1,1]$ & $(0.382,0.297,0.173,0.074,0.074)$ & $(0.379,0.293,0.174,0.077,0.077)$\\
$[2;1,1,0,0,0]$ & $(0.500,0.500,0.000,0.000,0.000)$ & $(0.500,0.500,0.000,0.000,0.000)$ & $[5;4,3,2,2,1]$ & $(0.324,0.258,0.169,0.169,0.079)$ & $(0.326,0.256,0.169,0.169,0.080)$\\
$[2;2,1,1,0,0]$ & $(0.611,0.194,0.194,0.000,0.000)$ & $(0.583,0.208,0.208,0.000,0.000)$ & $[5;5,2,2,1,1]$ & $(0.555,0.172,0.172,0.050,0.050)$ & $(0.538,0.174,0.174,0.057,0.057)$\\
$[2;2,1,1,1,0]$ & $(0.479,0.174,0.174,0.174,0.000)$ & $(0.463,0.179,0.179,0.179,0.000)$ & $[5;5,3,2,1,1]$ & $(0.518,0.247,0.138,0.048,0.048)$ & $(0.501,0.247,0.143,0.054,0.054)$\\
$[2;2,1,1,1,1]$ & $(0.397,0.151,0.151,0.151,0.151)$ & $(0.387,0.153,0.153,0.153,0.153)$ & $[5;5,3,2,2,1]$ & $(0.478,0.211,0.134,0.134,0.043)$ & $(0.463,0.214,0.138,0.138,0.049)$\\
$[2;1,1,1,0,0]$ & $(0.333,0.333,0.333,0.000,0.000)$ & $(0.333,0.333,0.333,0.000,0.000)$ & $[5;2,2,2,1,1]$ & $(0.256,0.256,0.256,0.116,0.116)$ & $(0.255,0.255,0.255,0.117,0.117)$\\
$[2;2,2,1,1,0]$ & $(0.396,0.396,0.104,0.104,0.000)$ & $(0.383,0.383,0.117,0.117,0.000)$ & $[5;3,3,2,1,1]$ & $(0.319,0.319,0.200,0.081,0.081)$ & $(0.316,0.316,0.200,0.084,0.084)$\\
$[2;2,2,1,1,1]$ & $(0.340,0.340,0.107,0.107,0.107)$ & $(0.331,0.331,0.113,0.113,0.113)$ & $[5;3,3,2,2,1]$ & $(0.286,0.286,0.189,0.189,0.050)$ & $(0.284,0.284,0.188,0.188,0.057)$\\
$[2;1,1,1,1,0]$ & $(0.250,0.250,0.250,0.250,0.000)$ & $(0.250,0.250,0.250,0.250,0.000)$ & $[6;2,2,2,1,1]$ & $(0.249,0.249,0.249,0.127,0.127)$ & $(0.249,0.249,0.249,0.127,0.127)$\\
$[2;2,2,2,1,1]$ & $(0.290,0.290,0.290,0.065,0.065)$ & $(0.283,0.283,0.283,0.075,0.075)$ & $[6;2,2,1,1,1]$ & $(0.340,0.340,0.107,0.107,0.107)$ & $(0.331,0.331,0.113,0.113,0.113)$\\
$[2;1,1,1,1,1]$ & $(0.200,0.200,0.200,0.200,0.200)$ & $(0.200,0.200,0.200,0.200,0.200)$ & $[6;3,2,1,1,1]$ & $(0.457,0.237,0.102,0.102,0.102)$ & $(0.443,0.239,0.106,0.106,0.106)$\\
$[3;1,1,1,0,0]$ & $(0.333,0.333,0.333,0.000,0.000)$ & $(0.333,0.333,0.333,0.000,0.000)$ & $[6;4,2,1,1,1]$ & $(0.532,0.214,0.085,0.085,0.085)$ & $(0.517,0.215,0.089,0.089,0.089)$\\
$[3;2,1,1,0,0]$ & $(0.611,0.194,0.194,0.000,0.000)$ & $(0.583,0.208,0.208,0.000,0.000)$ & $[6;3,3,1,1,1]$ & $(0.353,0.353,0.098,0.098,0.098)$ & $(0.350,0.350,0.100,0.100,0.100)$\\
$[3;2,1,1,1,0]$ & $(0.438,0.188,0.188,0.188,0.000)$ & $(0.430,0.190,0.190,0.190,0.000)$ & $[6;3,3,2,1,1]$ & $(0.319,0.319,0.200,0.081,0.081)$ & $(0.316,0.316,0.200,0.084,0.084)$\\
$[3;2,1,1,1,1]$ & $(0.345,0.164,0.164,0.164,0.164)$ & $(0.343,0.164,0.164,0.164,0.164)$ & $[6;3,3,2,2,2]$ & $(0.247,0.247,0.169,0.169,0.169)$ & $(0.248,0.248,0.168,0.168,0.168)$\\
$[3;3,1,1,1,0]$ & $(0.600,0.133,0.133,0.133,0.000)$ & $(0.580,0.140,0.140,0.140,0.000)$ & $[6;3,2,2,1,1]$ & $(0.333,0.221,0.221,0.112,0.112)$ & $(0.333,0.221,0.221,0.112,0.112)$\\
$[3;3,2,1,1,0]$ & $(0.535,0.270,0.098,0.098,0.000)$ & $(0.513,0.273,0.107,0.107,0.000)$ & $[6;4,2,2,1,1]$ & $(0.424,0.209,0.209,0.079,0.079)$ & $(0.420,0.207,0.207,0.083,0.083)$\\
$[3;3,2,1,1,1]$ & $(0.457,0.237,0.102,0.102,0.102)$ & $(0.443,0.239,0.106,0.106,0.106)$ & $[6;3,2,2,2,1]$ & $(0.300,0.198,0.198,0.198,0.104)$ & $(0.300,0.199,0.199,0.199,0.103)$\\
$[3;3,1,1,1,1]$ & $(0.502,0.125,0.125,0.125,0.125)$ & $(0.489,0.128,0.128,0.128,0.128)$ & $[6;4,3,3,1,1]$ & $(0.367,0.261,0.261,0.056,0.056)$ & $(0.361,0.259,0.259,0.060,0.060)$\\
$[3;3,2,2,1,1]$ & $(0.424,0.198,0.198,0.090,0.090)$ & $(0.409,0.202,0.202,0.093,0.093)$ & $[6;4,3,3,2,1]$ & $(0.299,0.238,0.238,0.150,0.075)$ & $(0.300,0.237,0.237,0.151,0.076)$\\
$[3;1,1,1,1,0]$ & $(0.250,0.250,0.250,0.250,0.000)$ & $(0.250,0.250,0.250,0.250,0.000)$ & $[6;4,3,2,2,1]$ & $(0.354,0.275,0.152,0.152,0.067)$ & $(0.350,0.271,0.154,0.154,0.070)$\\
$[3;2,2,1,1,0]$ & $(0.346,0.346,0.154,0.154,0.000)$ & $(0.343,0.343,0.157,0.157,0.000)$ & $[6;5,2,2,2,1]$ & $(0.449,0.169,0.169,0.169,0.045)$ & $(0.444,0.168,0.168,0.168,0.052)$\\
$[3;2,2,1,1,1]$ & $(0.295,0.295,0.136,0.136,0.136)$ & $(0.294,0.294,0.138,0.138,0.138)$ & $[6;5,3,3,1,1]$ & $(0.416,0.245,0.245,0.047,0.047)$ & $(0.411,0.243,0.243,0.052,0.052)$\\
$[3;3,3,1,1,1]$ & $(0.390,0.390,0.073,0.073,0.073)$ & $(0.381,0.381,0.079,0.079,0.079)$ & $[6;5,4,2,2,1]$ & $(0.374,0.323,0.132,0.132,0.039)$ & $(0.371,0.317,0.134,0.134,0.045)$\\
$[3;3,3,2,1,1]$ & $(0.364,0.364,0.155,0.059,0.059)$ & $(0.353,0.353,0.163,0.065,0.065)$ & $[7;2,2,2,1,1]$ & $(0.290,0.290,0.290,0.065,0.065)$ & $(0.283,0.283,0.283,0.075,0.075)$\\
$[3;1,1,1,1,1]$ & $(0.200,0.200,0.200,0.200,0.200)$ & $(0.200,0.200,0.200,0.200,0.200)$ & $[7;3,2,2,1,1]$ & $(0.424,0.198,0.198,0.090,0.090)$ & $(0.409,0.202,0.202,0.093,0.093)$\\\hline\hline
\end{longtable}
\end{footnotesize}

\newpage

\begin{footnotesize}
\begin{longtable}[c]{cccccc}
\caption[]{\begin{minipage}[t]{8cm} Average representation indices for $n\leq 5$ (cont.).\end{minipage}}\label{tabA2a_cont}\endhead\hline\hline
Game & AWI & ARI & Game & AWI & ARI\\\hline
$[3;2,2,2,1,1]$ & $(0.249,0.249,0.249,0.127,0.127)$ & $(0.249,0.249,0.249,0.127,0.127)$ & $[7;3,3,1,1,1]$ & $(0.390,0.390,0.073,0.073,0.073)$ & $(0.381,0.381,0.079,0.079,0.079)$\\
$[4;1,1,1,1,0]$ & $(0.250,0.250,0.250,0.250,0.000)$ & $(0.250,0.250,0.250,0.250,0.000)$ & $[7;3,3,2,2,1]$ & $(0.286,0.286,0.189,0.189,0.050)$ & $(0.284,0.284,0.188,0.188,0.057)$\\
$[4;2,2,1,1,0]$ & $(0.346,0.346,0.154,0.154,0.000)$ & $(0.343,0.343,0.157,0.157,0.000)$ & $[7;4,3,1,1,1]$ & $(0.495,0.290,0.072,0.072,0.072)$ & $(0.479,0.292,0.076,0.076,0.076)$\\
$[4;2,2,1,1,1]$ & $(0.300,0.300,0.133,0.133,0.133)$ & $(0.298,0.298,0.135,0.135,0.135)$ & $[7;4,3,2,1,1]$ & $(0.382,0.297,0.173,0.074,0.074)$ & $(0.379,0.293,0.174,0.077,0.077)$\\
$[4;2,1,1,1,0]$ & $(0.479,0.174,0.174,0.174,0.000)$ & $(0.463,0.179,0.179,0.179,0.000)$ & $[7;4,3,2,2,1]$ & $(0.354,0.275,0.152,0.152,0.067)$ & $(0.350,0.271,0.154,0.154,0.070)$\\
$[4;2,1,1,1,1]$ & $(0.345,0.164,0.164,0.164,0.164)$ & $(0.343,0.164,0.164,0.164,0.164)$ & $[7;3,2,2,2,1]$ & $(0.310,0.212,0.212,0.212,0.053)$ & $(0.308,0.210,0.210,0.210,0.061)$\\
$[4;3,1,1,1,0]$ & $(0.600,0.133,0.133,0.133,0.000)$ & $(0.580,0.140,0.140,0.140,0.000)$ & $[7;4,2,2,1,1]$ & $(0.488,0.177,0.177,0.078,0.078)$ & $(0.474,0.181,0.181,0.082,0.082)$\\
$[4;3,2,2,1,0]$ & $(0.402,0.258,0.258,0.081,0.000)$ & $(0.397,0.257,0.257,0.090,0.000)$ & $[7;5,2,2,1,1]$ & $(0.555,0.172,0.172,0.050,0.050)$ & $(0.538,0.174,0.174,0.057,0.057)$\\
$[4;3,1,1,1,1]$ & $(0.467,0.133,0.133,0.133,0.133)$ & $(0.460,0.135,0.135,0.135,0.135)$ & $[7;4,3,3,1,1]$ & $(0.367,0.261,0.261,0.056,0.056)$ & $(0.361,0.259,0.259,0.060,0.060)$\\
$[4;3,2,2,1,1]$ & $(0.333,0.221,0.221,0.112,0.112)$ & $(0.333,0.221,0.221,0.112,0.112)$ & $[7;4,3,3,2,2]$ & $(0.273,0.218,0.218,0.146,0.146)$ & $(0.275,0.217,0.217,0.145,0.145)$\\
$[4;3,2,1,1,1]$ & $(0.391,0.259,0.117,0.117,0.117)$ & $(0.388,0.257,0.118,0.118,0.118)$ & $[7;5,2,2,2,1]$ & $(0.449,0.169,0.169,0.169,0.045)$ & $(0.444,0.168,0.168,0.168,0.052)$\\
$[4;3,2,2,2,1]$ & $(0.310,0.212,0.212,0.212,0.053)$ & $(0.308,0.210,0.210,0.210,0.061)$ & $[7;5,3,3,2,1]$ & $(0.392,0.225,0.225,0.117,0.041)$ & $(0.386,0.223,0.223,0.121,0.046)$\\
$[4;4,1,1,1,1]$ & $(0.586,0.104,0.104,0.104,0.104)$ & $(0.571,0.107,0.107,0.107,0.107)$ & $[7;5,4,3,2,1]$ & $(0.354,0.305,0.190,0.114,0.037)$ & $(0.350,0.299,0.191,0.117,0.042)$\\
$[4;4,2,2,1,1]$ & $(0.488,0.177,0.177,0.078,0.078)$ & $(0.474,0.181,0.181,0.082,0.082)$ & $[7;3,3,2,2,2]$ & $(0.247,0.247,0.169,0.169,0.169)$ & $(0.248,0.248,0.168,0.168,0.168)$\\
$[4;4,2,1,1,1]$ & $(0.532,0.214,0.085,0.085,0.085)$ & $(0.517,0.215,0.089,0.089,0.089)$ & $[8;3,3,2,1,1]$ & $(0.364,0.364,0.155,0.059,0.059)$ & $(0.353,0.353,0.163,0.065,0.065)$\\
$[4;4,3,1,1,1]$ & $(0.495,0.290,0.072,0.072,0.072)$ & $(0.479,0.292,0.076,0.076,0.076)$ & $[8;3,3,2,2,1]$ & $(0.276,0.276,0.182,0.182,0.084)$ & $(0.275,0.275,0.182,0.182,0.085)$\\
$[4;4,3,2,2,1]$ & $(0.440,0.224,0.145,0.145,0.045)$ & $(0.423,0.228,0.149,0.149,0.052)$ & $[8;4,3,2,2,1]$ & $(0.324,0.258,0.169,0.169,0.079)$ & $(0.326,0.256,0.169,0.169,0.080)$\\
$[4;1,1,1,1,1]$ & $(0.200,0.200,0.200,0.200,0.200)$ & $(0.200,0.200,0.200,0.200,0.200)$ & $[8;5,3,2,1,1]$ & $(0.518,0.247,0.138,0.048,0.048)$ & $(0.501,0.247,0.143,0.054,0.054)$\\
$[4;2,2,2,1,1]$ & $(0.256,0.256,0.256,0.116,0.116)$ & $(0.255,0.255,0.255,0.117,0.117)$ & $[8;4,3,3,2,1]$ & $(0.299,0.238,0.238,0.150,0.075)$ & $(0.300,0.237,0.237,0.151,0.076)$\\
$[4;3,3,1,1,1]$ & $(0.353,0.353,0.098,0.098,0.098)$ & $(0.350,0.350,0.100,0.100,0.100)$ & $[8;5,3,3,2,1]$ & $(0.392,0.225,0.225,0.117,0.041)$ & $(0.386,0.223,0.223,0.121,0.046)$\\
$[4;3,3,2,2,1]$ & $(0.276,0.276,0.182,0.182,0.084)$ & $(0.275,0.275,0.182,0.182,0.085)$ & $[8;5,3,3,1,1]$ & $(0.416,0.245,0.245,0.047,0.047)$ & $(0.411,0.243,0.243,0.052,0.052)$\\
$[5;1,1,1,1,1]$ & $(0.200,0.200,0.200,0.200,0.200)$ & $(0.200,0.200,0.200,0.200,0.200)$ & $[8;5,4,3,2,2]$ & $(0.328,0.285,0.169,0.109,0.109)$ & $(0.325,0.279,0.172,0.112,0.112)$\\
$[5;2,2,1,1,0]$ & $(0.396,0.396,0.104,0.104,0.000)$ & $(0.383,0.383,0.117,0.117,0.000)$ & $[8;4,3,3,2,2]$ & $(0.273,0.218,0.218,0.146,0.146)$ & $(0.275,0.217,0.217,0.145,0.145)$\\
$[5;2,2,1,1,1]$ & $(0.295,0.295,0.136,0.136,0.136)$ & $(0.294,0.294,0.138,0.138,0.138)$ & $[9;4,3,2,2,1]$ & $(0.440,0.224,0.145,0.145,0.045)$ & $(0.423,0.228,0.149,0.149,0.052)$\\
$[5;3,2,1,1,0]$ & $(0.535,0.270,0.098,0.098,0.000)$ & $(0.513,0.273,0.107,0.107,0.000)$ & $[9;5,3,2,2,1]$ & $(0.478,0.211,0.134,0.134,0.043)$ & $(0.463,0.214,0.138,0.138,0.049)$\\
$[5;3,2,1,1,1]$ & $(0.391,0.259,0.117,0.117,0.117)$ & $(0.388,0.257,0.118,0.118,0.118)$ & $[9;5,4,2,2,1]$ & $(0.374,0.323,0.132,0.132,0.039)$ & $(0.371,0.317,0.134,0.134,0.045)$\\
$[5;2,1,1,1,1]$ & $(0.397,0.151,0.151,0.151,0.151)$ & $(0.387,0.153,0.153,0.153,0.153)$ & $[9;5,4,3,2,1]$ & $(0.354,0.305,0.190,0.114,0.037)$ & $(0.350,0.299,0.191,0.117,0.042)$\\
$[5;3,2,2,1,1]$ & $(0.367,0.233,0.233,0.083,0.083)$ & $(0.361,0.231,0.231,0.088,0.088)$ & $[9;5,4,3,2,2]$ & $(0.328,0.285,0.169,0.109,0.109)$ & $(0.325,0.279,0.172,0.112,0.112)$\\
$[5;3,2,2,1,0]$ & $(0.402,0.258,0.258,0.081,0.000)$ & $(0.397,0.257,0.257,0.090,0.000)$ & & &\\\hline\hline
\end{longtable}
\end{footnotesize}

\newpage

\begin{footnotesize}
\begin{longtable}[c]{cccccc}
\caption[]{\begin{minipage}[t]{12cm} Average representation type preserving indices for $n\leq 5$.\end{minipage}}\label{tabA2b}\endhead\hline\hline
Game & AWTI & ARTI & Game & AWTI & ARTI\\\hline
$[1;1,0,0,0,0]$ & $(1.000,0.000,0.000,0.000,0.000)$ & $(1.000,0.000,0.000,0.000,0.000)$ & $[5;3,1,1,1,1]$ & $(0.667,0.083,0.083,0.083,0.083)$ & $(0.587,0.103,0.103,0.103,0.103)$\\
$[1;1,1,0,0,0]$ & $(0.500,0.500,0.000,0.000,0.000)$ & $(0.500,0.500,0.000,0.000,0.000)$ & $[5;3,2,2,2,1]$ & $(0.294,0.206,0.206,0.206,0.089)$ & $(0.296,0.204,0.204,0.204,0.092)$\\
$[1;1,1,1,0,0]$ & $(0.333,0.333,0.333,0.000,0.000)$ & $(0.333,0.333,0.333,0.000,0.000)$ & $[5;4,1,1,1,1]$ & $(0.714,0.071,0.071,0.071,0.071)$ & $(0.643,0.089,0.089,0.089,0.089)$\\
$[1;1,1,1,1,0]$ & $(0.250,0.250,0.250,0.250,0.000)$ & $(0.250,0.250,0.250,0.250,0.000)$ & $[5;4,2,2,1,1]$ & $(0.407,0.198,0.198,0.099,0.099)$ & $(0.405,0.199,0.199,0.099,0.099)$\\
$[1;1,1,1,1,1]$ & $(0.200,0.200,0.200,0.200,0.200)$ & $(0.200,0.200,0.200,0.200,0.200)$ & $[5;4,3,2,1,1]$ & $(0.378,0.321,0.166,0.068,0.068)$ & $(0.375,0.311,0.169,0.072,0.072)$\\
$[2;1,1,0,0,0]$ & $(0.500,0.500,0.000,0.000,0.000)$ & $(0.500,0.500,0.000,0.000,0.000)$ & $[5;4,3,2,2,1]$ & $(0.322,0.239,0.170,0.170,0.099)$ & $(0.324,0.241,0.169,0.169,0.096)$\\
$[2;2,1,1,0,0]$ & $(0.667,0.167,0.167,0.000,0.000)$ & $(0.611,0.194,0.194,0.000,0.000)$ & $[5;5,2,2,1,1]$ & $(0.625,0.137,0.137,0.051,0.051)$ & $(0.582,0.148,0.148,0.061,0.061)$\\
$[2;2,1,1,1,0]$ & $(0.625,0.125,0.125,0.125,0.000)$ & $(0.550,0.150,0.150,0.150,0.000)$ & $[5;5,3,2,1,1]$ & $(0.548,0.258,0.123,0.035,0.035)$ & $(0.522,0.257,0.132,0.045,0.045)$\\
$[2;2,1,1,1,1]$ & $(0.600,0.100,0.100,0.100,0.100)$ & $(0.511,0.122,0.122,0.122,0.122)$ & $[5;5,3,2,2,1]$ & $(0.514,0.180,0.126,0.126,0.054)$ & $(0.488,0.190,0.132,0.132,0.058)$\\
$[2;1,1,1,0,0]$ & $(0.333,0.333,0.333,0.000,0.000)$ & $(0.333,0.333,0.333,0.000,0.000)$ & $[5;2,2,2,1,1]$ & $(0.267,0.267,0.267,0.100,0.100)$ & $(0.261,0.261,0.261,0.108,0.108)$\\
$[2;2,2,1,1,0]$ & $(0.375,0.375,0.125,0.125,0.000)$ & $(0.361,0.361,0.139,0.139,0.000)$ & $[5;3,3,2,1,1]$ & $(0.356,0.356,0.175,0.056,0.056)$ & $(0.342,0.342,0.182,0.067,0.067)$\\
$[2;2,2,1,1,1]$ & $(0.350,0.350,0.100,0.100,0.100)$ & $(0.329,0.329,0.114,0.114,0.114)$ & $[5;3,3,2,2,1]$ & $(0.265,0.265,0.193,0.193,0.084)$ & $(0.267,0.267,0.190,0.190,0.086)$\\
$[2;1,1,1,1,0]$ & $(0.250,0.250,0.250,0.250,0.000)$ & $(0.250,0.250,0.250,0.250,0.000)$ & $[6;2,2,2,1,1]$ & $(0.267,0.267,0.267,0.100,0.100)$ & $(0.261,0.261,0.261,0.108,0.108)$\\
$[2;2,2,2,1,1]$ & $(0.267,0.267,0.267,0.100,0.100)$ & $(0.261,0.261,0.261,0.108,0.108)$ & $[6;2,2,1,1,1]$ & $(0.350,0.350,0.100,0.100,0.100)$ & $(0.329,0.329,0.114,0.114,0.114)$\\
$[2;1,1,1,1,1]$ & $(0.200,0.200,0.200,0.200,0.200)$ & $(0.200,0.200,0.200,0.200,0.200)$ & $[6;3,2,1,1,1]$ & $(0.540,0.173,0.096,0.096,0.096)$ & $(0.499,0.192,0.103,0.103,0.103)$\\
$[3;1,1,1,0,0]$ & $(0.333,0.333,0.333,0.000,0.000)$ & $(0.333,0.333,0.333,0.000,0.000)$ & $[6;4,2,1,1,1]$ & $(0.620,0.224,0.052,0.052,0.052)$ & $(0.576,0.223,0.067,0.067,0.067)$\\
$[3;2,1,1,0,0]$ & $(0.667,0.167,0.167,0.000,0.000)$ & $(0.611,0.194,0.194,0.000,0.000)$ & $[6;3,3,1,1,1]$ & $(0.393,0.393,0.071,0.071,0.071)$ & $(0.373,0.373,0.085,0.085,0.085)$\\
$[3;2,1,1,1,0]$ & $(0.375,0.208,0.208,0.208,0.000)$ & $(0.383,0.206,0.206,0.206,0.000)$ & $[6;3,3,2,1,1]$ & $(0.356,0.356,0.175,0.056,0.056)$ & $(0.342,0.342,0.182,0.067,0.067)$\\
$[3;2,1,1,1,1]$ & $(0.314,0.171,0.171,0.171,0.171)$ & $(0.321,0.170,0.170,0.170,0.170)$ & $[6;3,3,2,2,2]$ & $(0.243,0.243,0.171,0.171,0.171)$ & $(0.245,0.245,0.170,0.170,0.170)$\\
$[3;3,1,1,1,0]$ & $(0.700,0.100,0.100,0.100,0.000)$ & $(0.633,0.122,0.122,0.122,0.000)$ & $[6;3,2,2,1,1]$ & $(0.311,0.233,0.233,0.111,0.111)$ & $(0.317,0.231,0.231,0.111,0.111)$\\
$[3;3,2,1,1,0]$ & $(0.602,0.249,0.075,0.075,0.000)$ & $(0.558,0.258,0.092,0.092,0.000)$ & $[6;4,2,2,1,1]$ & $(0.407,0.198,0.198,0.099,0.099)$ & $(0.405,0.199,0.199,0.099,0.099)$\\
$[3;3,2,1,1,1]$ & $(0.540,0.173,0.096,0.096,0.096)$ & $(0.499,0.192,0.103,0.103,0.103)$ & $[6;3,2,2,2,1]$ & $(0.294,0.206,0.206,0.206,0.089)$ & $(0.296,0.204,0.204,0.204,0.092)$\\
$[3;3,1,1,1,1]$ & $(0.667,0.083,0.083,0.083,0.083)$ & $(0.587,0.103,0.103,0.103,0.103)$ & $[6;4,3,3,1,1]$ & $(0.361,0.269,0.269,0.051,0.051)$ & $(0.354,0.264,0.264,0.059,0.059)$\\
$[3;3,2,2,1,1]$ & $(0.528,0.176,0.176,0.060,0.060)$ & $(0.479,0.188,0.188,0.073,0.073)$ & $[6;4,3,3,2,1]$ & $(0.302,0.244,0.244,0.152,0.059)$ & $(0.303,0.241,0.241,0.152,0.062)$\\
$[3;1,1,1,1,0]$ & $(0.250,0.250,0.250,0.250,0.000)$ & $(0.250,0.250,0.250,0.250,0.000)$ & $[6;4,3,2,2,1]$ & $(0.359,0.288,0.148,0.148,0.057)$ & $(0.354,0.280,0.152,0.152,0.062)$\\
$[3;2,2,1,1,0]$ & $(0.375,0.375,0.125,0.125,0.000)$ & $(0.361,0.361,0.139,0.139,0.000)$ & $[6;5,2,2,2,1]$ & $(0.415,0.170,0.170,0.170,0.075)$ & $(0.415,0.169,0.169,0.169,0.077)$\\
$[3;2,2,1,1,1]$ & $(0.267,0.267,0.156,0.156,0.156)$ & $(0.273,0.273,0.151,0.151,0.151)$ & $[6;5,3,3,1,1]$ & $(0.394,0.258,0.258,0.045,0.045)$ & $(0.392,0.251,0.251,0.053,0.053)$\\
$[3;3,3,1,1,1]$ & $(0.393,0.393,0.071,0.071,0.071)$ & $(0.373,0.373,0.085,0.085,0.085)$ & $[6;5,4,2,2,1]$ & $(0.384,0.321,0.123,0.123,0.049)$ & $(0.379,0.314,0.127,0.127,0.054)$\\
$[3;3,3,2,1,1]$ & $(0.356,0.356,0.175,0.056,0.056)$ & $(0.342,0.342,0.182,0.067,0.067)$ & $[7;2,2,2,1,1]$ & $(0.267,0.267,0.267,0.100,0.100)$ & $(0.261,0.261,0.261,0.108,0.108)$\\
$[3;1,1,1,1,1]$ & $(0.200,0.200,0.200,0.200,0.200)$ & $(0.200,0.200,0.200,0.200,0.200)$ & $[7;3,2,2,1,1]$ & $(0.528,0.176,0.176,0.060,0.060)$ & $(0.479,0.188,0.188,0.073,0.073)$\\\hline\hline
\end{longtable}
\end{footnotesize}

\newpage

\begin{footnotesize}
\begin{longtable}[c]{cccccc}
\caption[]{\begin{minipage}[t]{12cm} Average representation type preserving indices for $n\leq 5$ (cont.).\end{minipage}}\label{tabA2b_cont}\endhead\hline\hline
Game & AWTI & ARTI & Game & AWTI & ARTI\\\hline
$[3;2,2,2,1,1]$ & $(0.267,0.267,0.267,0.100,0.100)$ & $(0.261,0.261,0.261,0.108,0.108)$ & $[7;3,3,1,1,1]$ & $(0.393,0.393,0.071,0.071,0.071)$ & $(0.373,0.373,0.085,0.085,0.085)$\\
$[4;1,1,1,1,0]$ & $(0.250,0.250,0.250,0.250,0.000)$ & $(0.250,0.250,0.250,0.250,0.000)$ & $[7;3,3,2,2,1]$ & $(0.265,0.265,0.193,0.193,0.084)$ & $(0.267,0.267,0.190,0.190,0.086)$\\
$[4;2,2,1,1,0]$ & $(0.375,0.375,0.125,0.125,0.000)$ & $(0.361,0.361,0.139,0.139,0.000)$ & $[7;4,3,1,1,1]$ & $(0.603,0.261,0.045,0.045,0.045)$ & $(0.552,0.271,0.059,0.059,0.059)$\\
$[4;2,2,1,1,1]$ & $(0.375,0.375,0.083,0.083,0.083)$ & $(0.345,0.345,0.103,0.103,0.103)$ & $[7;4,3,2,1,1]$ & $(0.378,0.321,0.166,0.068,0.068)$ & $(0.375,0.311,0.169,0.072,0.072)$\\
$[4;2,1,1,1,0]$ & $(0.625,0.125,0.125,0.125,0.000)$ & $(0.550,0.150,0.150,0.150,0.000)$ & $[7;4,3,2,2,1]$ & $(0.359,0.288,0.148,0.148,0.057)$ & $(0.354,0.280,0.152,0.152,0.062)$\\
$[4;2,1,1,1,1]$ & $(0.314,0.171,0.171,0.171,0.171)$ & $(0.321,0.170,0.170,0.170,0.170)$ & $[7;3,2,2,2,1]$ & $(0.294,0.206,0.206,0.206,0.089)$ & $(0.296,0.204,0.204,0.204,0.092)$\\
$[4;3,1,1,1,0]$ & $(0.700,0.100,0.100,0.100,0.000)$ & $(0.633,0.122,0.122,0.122,0.000)$ & $[7;4,2,2,1,1]$ & $(0.556,0.167,0.167,0.056,0.056)$ & $(0.517,0.175,0.175,0.067,0.067)$\\
$[4;3,2,2,1,0]$ & $(0.369,0.261,0.261,0.108,0.000)$ & $(0.371,0.258,0.258,0.113,0.000)$ & $[7;5,2,2,1,1]$ & $(0.625,0.137,0.137,0.051,0.051)$ & $(0.582,0.148,0.148,0.061,0.061)$\\
$[4;3,1,1,1,1]$ & $(0.417,0.146,0.146,0.146,0.146)$ & $(0.421,0.145,0.145,0.145,0.145)$ & $[7;4,3,3,1,1]$ & $(0.361,0.269,0.269,0.051,0.051)$ & $(0.354,0.264,0.264,0.059,0.059)$\\
$[4;3,2,2,1,1]$ & $(0.311,0.233,0.233,0.111,0.111)$ & $(0.317,0.231,0.231,0.111,0.111)$ & $[7;4,3,3,2,2]$ & $(0.277,0.208,0.208,0.153,0.153)$ & $(0.279,0.210,0.210,0.151,0.151)$\\
$[4;3,2,1,1,1]$ & $(0.397,0.294,0.103,0.103,0.103)$ & $(0.391,0.283,0.109,0.109,0.109)$ & $[7;5,2,2,2,1]$ & $(0.415,0.170,0.170,0.170,0.075)$ & $(0.415,0.169,0.169,0.169,0.077)$\\
$[4;3,2,2,2,1]$ & $(0.294,0.206,0.206,0.206,0.089)$ & $(0.296,0.204,0.204,0.204,0.092)$ & $[7;5,3,3,2,1]$ & $(0.365,0.219,0.219,0.146,0.052)$ & $(0.363,0.218,0.218,0.146,0.056)$\\
$[4;4,1,1,1,1]$ & $(0.714,0.071,0.071,0.071,0.071)$ & $(0.643,0.089,0.089,0.089,0.089)$ & $[7;5,4,3,2,1]$ & $(0.354,0.305,0.190,0.114,0.037)$ & $(0.350,0.299,0.191,0.117,0.042)$\\
$[4;4,2,2,1,1]$ & $(0.556,0.167,0.167,0.056,0.056)$ & $(0.517,0.175,0.175,0.067,0.067)$ & $[7;3,3,2,2,2]$ & $(0.243,0.243,0.171,0.171,0.171)$ & $(0.245,0.245,0.170,0.170,0.170)$\\
$[4;4,2,1,1,1]$ & $(0.620,0.224,0.052,0.052,0.052)$ & $(0.576,0.223,0.067,0.067,0.067)$ & $[8;3,3,2,1,1]$ & $(0.356,0.356,0.175,0.056,0.056)$ & $(0.342,0.342,0.182,0.067,0.067)$\\
$[4;4,3,1,1,1]$ & $(0.603,0.261,0.045,0.045,0.045)$ & $(0.552,0.271,0.059,0.059,0.059)$ & $[8;3,3,2,2,1]$ & $(0.265,0.265,0.193,0.193,0.084)$ & $(0.267,0.267,0.190,0.190,0.086)$\\
$[4;4,3,2,2,1]$ & $(0.467,0.196,0.140,0.140,0.057)$ & $(0.440,0.207,0.145,0.145,0.062)$ & $[8;4,3,2,2,1]$ & $(0.322,0.239,0.170,0.170,0.099)$ & $(0.324,0.241,0.169,0.169,0.096)$\\
$[4;1,1,1,1,1]$ & $(0.200,0.200,0.200,0.200,0.200)$ & $(0.200,0.200,0.200,0.200,0.200)$ & $[8;5,3,2,1,1]$ & $(0.548,0.258,0.123,0.035,0.035)$ & $(0.522,0.257,0.132,0.045,0.045)$\\
$[4;2,2,2,1,1]$ & $(0.267,0.267,0.267,0.100,0.100)$ & $(0.261,0.261,0.261,0.108,0.108)$ & $[8;4,3,3,2,1]$ & $(0.302,0.244,0.244,0.152,0.059)$ & $(0.303,0.241,0.241,0.152,0.062)$\\
$[4;3,3,1,1,1]$ & $(0.393,0.393,0.071,0.071,0.071)$ & $(0.373,0.373,0.085,0.085,0.085)$ & $[8;5,3,3,2,1]$ & $(0.365,0.219,0.219,0.146,0.052)$ & $(0.363,0.218,0.218,0.146,0.056)$\\
$[4;3,3,2,2,1]$ & $(0.265,0.265,0.193,0.193,0.084)$ & $(0.267,0.267,0.190,0.190,0.086)$ & $[8;5,3,3,1,1]$ & $(0.394,0.258,0.258,0.045,0.045)$ & $(0.392,0.251,0.251,0.053,0.053)$\\
$[5;1,1,1,1,1]$ & $(0.200,0.200,0.200,0.200,0.200)$ & $(0.200,0.200,0.200,0.200,0.200)$ & $[8;5,4,3,2,2]$ & $(0.347,0.294,0.149,0.105,0.105)$ & $(0.340,0.285,0.156,0.109,0.109)$\\
$[5;2,2,1,1,0]$ & $(0.375,0.375,0.125,0.125,0.000)$ & $(0.361,0.361,0.139,0.139,0.000)$ & $[8;4,3,3,2,2]$ & $(0.277,0.208,0.208,0.153,0.153)$ & $(0.279,0.210,0.210,0.151,0.151)$\\
$[5;2,2,1,1,1]$ & $(0.267,0.267,0.156,0.156,0.156)$ & $(0.273,0.273,0.151,0.151,0.151)$ & $[9;4,3,2,2,1]$ & $(0.467,0.196,0.140,0.140,0.057)$ & $(0.440,0.207,0.145,0.145,0.062)$\\
$[5;3,2,1,1,0]$ & $(0.602,0.249,0.075,0.075,0.000)$ & $(0.558,0.258,0.092,0.092,0.000)$ & $[9;5,3,2,2,1]$ & $(0.514,0.180,0.126,0.126,0.054)$ & $(0.488,0.190,0.132,0.132,0.058)$\\
$[5;3,2,1,1,1]$ & $(0.397,0.294,0.103,0.103,0.103)$ & $(0.391,0.283,0.109,0.109,0.109)$ & $[9;5,4,2,2,1]$ & $(0.384,0.321,0.123,0.123,0.049)$ & $(0.379,0.314,0.127,0.127,0.054)$\\
$[5;2,1,1,1,1]$ & $(0.600,0.100,0.100,0.100,0.100)$ & $(0.511,0.122,0.122,0.122,0.122)$ & $[9;5,4,3,2,1]$ & $(0.354,0.305,0.190,0.114,0.037)$ & $(0.350,0.299,0.191,0.117,0.042)$\\
$[5;3,2,2,1,1]$ & $(0.361,0.222,0.222,0.097,0.097)$ & $(0.354,0.222,0.222,0.101,0.101)$ & $[9;5,4,3,2,2]$ & $(0.347,0.294,0.149,0.105,0.105)$ & $(0.340,0.285,0.156,0.109,0.109)$\\
$[5;3,2,2,1,0]$ & $(0.369,0.261,0.261,0.108,0.000)$ & $(0.371,0.258,0.258,0.113,0.000)$ & & &\\\hline\hline
\end{longtable}
\end{footnotesize}

\end{landscape}

\end{appendix}

\end{document}